\documentclass[12pt,draftclsnofoot,onecolumn]{IEEEtran}
\usepackage{amsmath,amssymb,amsthm}
\usepackage{cite}
\usepackage{bm}
\usepackage{microtype}
\usepackage[hidelinks]{hyperref}

\allowdisplaybreaks

\DeclareMathOperator*{\argmax}{arg\,max}
\DeclareMathOperator*{\argmin}{arg\,min}

\newcommand{\R}{\mathbb{R}}
\newcommand{\Pp}{\mathbb{P}}
\newcommand{\Sph}{\mathbb{S}}
\newcommand{\Normal}{\mathcal{N}}
\newcommand{\cA}{\mathcal{A}}
\newcommand{\cC}{\mathcal{C}}
\newcommand{\cE}{\mathcal{E}}
\newcommand{\cG}{\mathcal{G}}
\newcommand{\cI}{\mathcal{I}}
\newcommand{\cM}{\mathcal{M}}
\newcommand{\cP}{\mathcal{P}}
\newcommand{\cQ}{\mathcal{Q}}
\newcommand{\dd}{\,\mathrm{d}}
\newcommand{\Pe}{P_{\mathrm e}}
\newcommand{\Qfun}{Q}
\newcommand{\bc}{\bm c}
\newcommand{\be}{\bm e}
\newcommand{\bA}{\mathbf A}
\newcommand{\bG}{\mathbf G}

\newcommand{\bx}{\bm x}
\newcommand{\by}{\bm y}
\newcommand{\bY}{\bm Y}
\newcommand{\bv}{\bm v}
\newcommand{\ba}{\bm a}
\newcommand{\bI}{\mathbf I}
\newcommand{\bZ}{\bm Z}
\newcommand{\bu}{\bm u}
\newcommand{\bzero}{\bm 0}
\newcommand{\T}{\mathsf T}
\newcommand{\eps}{\varepsilon}

\newtheorem{theorem}{Theorem}
\newtheorem{proposition}{Proposition}
\newtheorem{lemma}{Lemma}
\newtheorem{corollary}{Corollary}
\newtheorem{conjecture}{Conjecture}
\theoremstyle{definition}

\title{Beyond Minimum Distance: The Optimal Leading Coefficient in the High-SNR Error-Probability Expansion for AWGN Spherical Codes}

\author{Nikola Zlatanov
\thanks{N.~Zlatanov is with Innopolis University, Innopolis, 420500, Russia (e-mail: n.zlatanov@innopolis.ru).}
}

\begin{document}
\maketitle

\begin{abstract}
 Packing-optimal \((M,n)\) spherical codes attain the largest achievable minimum distance and therefore achieve the optimal exponential decay rate of the error probability at high signal-to-noise ratio (SNR). Among packing-optimal codebooks held fixed as the SNR grows, the smallest leading coefficient is
\(K^{\mathrm{fix}}_{M,n}\), equal to the smallest number of ordered closest pairs of codewords among packing-optimal codebooks.  We show that SNR-wise codebook optimization
can achieve a smaller leading coefficient, and thereby a smaller error
probability, than any packing-optimal codebook held fixed.
Let \(\Pe^{\ast}(M,n;\gamma)\) be the SNR-wise minimum exact
maximum-likelihood (ML) error probability and
let \(\Pe^{\mathrm{fix}}(M,n;\gamma)\) be the minimum exact error over all
packing-optimal codebooks.  We prove that
\(\Pe^{\ast}=B_\gamma^{\ast}(K^{\ast}_{M,n}+o(1))\) and
\(\Pe^{\mathrm{fix}}
=B_\gamma^{\ast}(K^{\mathrm{fix}}_{M,n}+o(1))\),  and that \(K^{\ast}_{M,n}\leq K^{\mathrm{fix}}_{M,n}\), where
\(B_\gamma^{\ast}\) is the common Gaussian-tail factor.   Hence, whenever
\(K^{\ast}_{M,n}<K^{\mathrm{fix}}_{M,n}\), the SNR-wise minimum exact
error is strictly smaller than the best packing-optimal
benchmark at all sufficiently high SNRs, i.e., \(\Pe^{\ast}<\Pe^{\mathrm{fix}}\).
We also show that the minimum of the Gaussian soft-packing energy converges to \(K^{\ast}_{M,n}\).

The orthoplex-bound construction reveals how
\(K^{\ast}_{M,n}<K^{\mathrm{fix}}_{M,n}\) can occur.  The construction
splits the pairs that are closest in the limiting packing into two
classes.  At finite SNR, the pairs in the first class are deliberately made slightly closer than the optimal packing
distance. This change is too small to affect either the optimal asymptotic
error exponent or the selected pairs' unit contributions to the constructed
family's leading coefficient. However, this creates enough geometric freedom to move the other pairs---those in the second class---farther apart
on a larger but still vanishing scale, so that
their contributions to the constructed family's leading coefficient
vanish.  By contrast, every
closest pair of a packing-optimal codebook held fixed contributes one to
that codebook's fixed-codebook leading coefficient.

We apply the developed theoretical framework to the simplex range,
\(2\leq M\leq n+1\), and the orthoplex-bound range, \(M=n+k\),
\(2\leq k\leq n\).  We obtain
\(K^{\ast}_{M,n}=K^{\mathrm{fix}}_{M,n}=M(M-1)\) for
\(2\leq M\leq n+1\).  For \(M=n+k\), \(2\leq k\leq n\), we prove
\(K^{\mathrm{fix}}_{n+k,n}=4n(k-1)\) and
\(K^{\ast}_{n+k,n}\leq4k(k-1)\).  For \(n\geq3\) and \(M=n+2\),
\(K^{\ast}_{n+2,n}=8<4n=K^{\mathrm{fix}}_{n+2,n}\), so the optimized
error is asymptotically smaller by the factor \(n/2\) than the best
packing-optimal error.  At the terminal endpoint \(k=n\),
\(K^{\ast}_{2n,n}=K^{\mathrm{fix}}_{2n,n}=4n(n-1)\).  For
\(3\leq k\leq n-1\), we conjecture that
\(K^{\ast}_{n+k,n}=4k(k-1)\).
\end{abstract}

\begin{IEEEkeywords}
AWGN channels, spherical codes, maximum-likelihood decoding, message error probability, high-SNR asymptotics, spherical packing, codebook optimization, leading coefficients.
\end{IEEEkeywords}

\section{Introduction}
\label{sec:introduction}

\IEEEPARstart{A}{t} high average signal-to-noise ratio (SNR) per real
channel use, finite spherical-code design raises two nested questions.
First, how rapidly can the minimum achievable average message error
probability under maximum-likelihood (ML) decoding decay?  Second, after
removing that decay scale, how small can the
remaining leading coefficient be?  Optimal spherical packing, through the
largest achievable minimum distance, answers the first question.  It need
not answer the second.  Codebook families can have the same
packing-determined exponential decay yet differ by a nonvanishing
multiplicative factor in error probability.

For a packing-optimal codebook held fixed as the SNR grows, classical
high-SNR theory shows that its number of ordered closest pairs of
codewords determines the leading coefficient in its high-SNR
error-probability expansion
\cite{Proakis2001,AlvaradoAgrellBrannstrom2018}.  Dividing this count by
\(M\) gives the codebook's average kissing number.  Hence, among
packing-optimal codebooks that are each held fixed as the SNR grows,
minimizing the leading coefficient is equivalent to choosing a codebook
with the smallest ordered closest-pair count, or, equivalently, the
smallest average kissing number.  The optimization studied here is
different: the codeword locations may be reoptimized at every SNR.
Because pairwise Gaussian error probabilities depend exponentially on
the SNR and the pairwise squared distances, vanishing changes in
codeword distances can still change the error probability by a
nonvanishing factor.  This allows an SNR-wise optimized codebook family
to approach the set of optimal packings while having a smaller leading
coefficient than every packing-optimal codebook held fixed.  Its error
probability is then smaller than the best fixed packing-optimal benchmark
at all sufficiently high SNRs.  Consequently, optimizing the codebook at
each SNR and then taking the high-SNR limit can yield a smaller leading
coefficient than first taking the high-SNR limit for each fixed codebook
and then optimizing over the packing-optimal class.

The explicit SNR-dependent construction developed in this paper for the
orthoplex-bound range \(M=n+k\), \(2\leq k<n\), makes this phenomenon
concrete.  The construction splits the codeword pairs that are closest in
the limiting packing into two classes.  At finite SNR, the pairs in the
first class are deliberately made slightly closer than the optimal packing
distance.
This deterioration is too small to change either the optimal asymptotic
error exponent or the selected pairs' unit contributions to the constructed
family's leading coefficient.  However, it creates enough geometric
freedom to move the other pairs---those in the second class---farther apart
on a larger but still vanishing scale, so that
their contributions to the constructed family's leading coefficient
vanish.  The constructed family retains the packing-optimal exponent and
approaches an optimal packing as the SNR grows, yet lies outside the set of
optimal packings at every sufficiently large finite SNR.  Its leading
coefficient is smaller than that of every packing-optimal codebook held
fixed as the SNR grows.   
This improvement cannot be obtained by holding any one codebook fixed.  The
smaller coefficient is
obtained only because the codebook continues to change with SNR.

Let \(\Pe^{\ast}(M,n;\gamma)\) denote the exact
average ML message error probability minimized over all \((M,n)\) spherical
codebooks at SNR
\(\gamma\).  For every fixed finite pair \((M,n)\) considered
here, we prove the general expansion
\begin{equation*}
  \Pe^{\ast}(M,n;\gamma)
  =B_\gamma^{\ast}\bigl(K^{\ast}_{M,n}+o(1)\bigr),
\end{equation*}
where \(B_\gamma^{\ast}\) contains the packing-determined exponential
decay and Gaussian prefactor, and \(K^{\ast}_{M,n}\) is a well-defined
optimal leading coefficient.

For comparison, let \(\Pe^{\mathrm{fix}}(M,n;\gamma)\) denote the exact
average ML message error probability minimized over all spherical
codebooks attaining the largest possible minimum distance.  Then
\begin{equation*}
  \Pe^{\mathrm{fix}}(M,n;\gamma)
  =B_\gamma^{\ast}\bigl(K^{\mathrm{fix}}_{M,n}+o(1)\bigr),
\end{equation*}
where \(K^{\mathrm{fix}}_{M,n}\) is the smallest number of ordered
closest-pair indices among packing-optimal codebooks and \(K^{\mathrm{fix}}_{M,n}/M\) is the
smallest average kissing number.  We show that
\(K^{\ast}_{M,n}\leq K^{\mathrm{fix}}_{M,n}\).  If
\(K^{\ast}_{M,n}<K^{\mathrm{fix}}_{M,n}\), then
\(\Pe^{\ast}(M,n;\gamma)<\Pe^{\mathrm{fix}}(M,n;\gamma)\) for all
sufficiently large \(\gamma\), and
\begin{equation*}
  \frac{\Pe^{\ast}(M,n;\gamma)}
       {\Pe^{\mathrm{fix}}(M,n;\gamma)}
  \longrightarrow
  \frac{K^{\ast}_{M,n}}{K^{\mathrm{fix}}_{M,n}}<1.
\end{equation*}
Thus the improvement is a persistent multiplicative reduction in the
minimum exact error probability.

We call an SNR-dependent family \emph{packing-tail competitive} if its
normalized error remains bounded on the packing-tail scale.  We show that
every such family approaches the set of optimal packings.  If such a family
has a leading coefficient, then, along any convergent
codebook subsequence, only pairs that are closest in the packing-optimal
limiting codebook can contribute to that coefficient.  The contribution of
each such pair depends on the rate at which its correlation approaches the
optimal packing value and may be zero, fractional, one, or greater than one.
For a family attaining the optimal
leading coefficient, these pairwise contributions sum to
\(K^{\ast}_{M,n}\).  Thus, for a family attaining the optimal leading
coefficient, the packing-optimal limiting
codebook identifies which pairs can contribute to \(K^{\ast}_{M,n}\),
whereas the SNR-dependent codebook optimization determines how much each
such pair contributes to \(K^{\ast}_{M,n}\).

To analyze this geometry, we introduce a Gaussian soft-packing energy
that continuously accounts for all pairwise confusions.  We prove that
it uniformly represents the normalized exact ML error for the
SNR-dependent families relevant to optimized high-SNR design, and that
its minimum converges to \(K^{\ast}_{M,n}\).  Consequently, every
family whose Gaussian soft-packing energy is within an additive \(o(1)\)
of its global minimum is asymptotically ML-optimal.

\subsection{Contributions}
\label{subsec:contributions}

The contributions are as follows.  The first three contributions apply to every fixed
finite \((M,n)\) (with \(M\geq2\)) for which the best packing has nonzero
minimum distance.  The last two apply the theory throughout the simplex
and orthoplex-bound ranges.
\begin{enumerate}
\item \emph{Optimal leading-coefficient characterization and fixed-codebook benchmark.}
We introduce a Gaussian soft-packing energy and prove that it uniformly
represents the normalized exact ML error on the bounded Gaussian
soft-packing energy classes
relevant to packing-tail-competitive designs.  Its minimum converges to a
well-defined optimal leading coefficient \(K^{\ast}_{M,n}\), yielding the
high-SNR expansion of the minimum average error probability,
\(\Pe^{\ast}(M,n;\gamma)
  =B_\gamma^{\ast}(K^{\ast}_{M,n}+o(1))\).  We also prove
\(\Pe^{\mathrm{fix}}(M,n;\gamma)
  =B_\gamma^{\ast}(K^{\mathrm{fix}}_{M,n}+o(1))\) for the best
packing-optimal benchmark.  We show that
\(K^{\ast}_{M,n}\leq K^{\mathrm{fix}}_{M,n}\).  Hence, when
\(K^{\ast}_{M,n}<K^{\mathrm{fix}}_{M,n}\),
\(\Pe^{\ast}(M,n;\gamma)<\Pe^{\mathrm{fix}}(M,n;\gamma)\) for all
sufficiently large \(\gamma\), and their ratio tends to
\(K^{\ast}_{M,n}/K^{\mathrm{fix}}_{M,n}\) as the SNR grows.

\item \emph{Localization and SNR-dependent closest-pair contributions of
packing-tail-competitive families.}
We prove that every family satisfying
\(\Pe(\cC_\gamma;\gamma)/B_\gamma^{\ast}=O(1)\) approaches the
packing-optimal set.  Along convergent subsequences, only pairs that are
closest in the packing-optimal limiting codebook can have nonvanishing
normalized contributions, and these contributions are determined by their
scaled correlation offsets.  Whenever the family's leading coefficient
exists, these contributions sum to that coefficient; for a family attaining
the optimal leading coefficient, they sum to \(K^{\ast}_{M,n}\).  This
explains why neither an optimal packing nor its
ordinary closest-pair count alone determines the optimal leading
coefficient \(K^{\ast}_{M,n}\).  Minimizing over all packing-optimal limiting codebooks and
all jointly feasible scaled closest-pair offsets yields an exact
variational characterization of \(K^{\ast}_{M,n}\).

\item \emph{Attainment and matching-bound certification.}
We prove that codebook families whose Gaussian soft-packing energy is
within an additive \(o(1)\) of its global minimum are asymptotically
ML-optimal.  We then give a matching-bounds criterion that certifies both
a proposed optimal leading coefficient and a full SNR-indexed
construction attaining the optimal leading coefficient.

\item \emph{Simplex code.}
For every \(2\leq M\leq n+1\), we prove that an embedded regular
\((M-1)\)-simplex minimizes the soft-packing objective exactly at every
positive SNR and that
\(K^{\ast}_{M,n}=K^{\mathrm{fix}}_{M,n}=M(M-1)\).  The constant simplex
family is therefore asymptotically ML-optimal under SNR-wise reoptimization.

\item \emph{Orthoplex-bound range.}
For the orthoplex-bound range \(M=n+k\),
\(2\leq k\leq n\), we determine
\(K^{\mathrm{fix}}_{n+k,n}=4n(k-1)\).  For every \(2\leq k<n\), an
explicit SNR-dependent family has leading coefficient \(4k(k-1)\),
which proves \(K^{\ast}_{n+k,n}\leq 4k(k-1)\).  Hence, the optimal leading
coefficient, \(K^{\ast}_{n+k,n}\), is strictly smaller than the fixed benchmark, \(K^{\mathrm{fix}}_{n+k,n}\), and
\begin{equation*}
  \lim_{\gamma\to\infty}
  \frac{\Pe^{\ast}(n+k,n;\gamma)}
       {\Pe^{\mathrm{fix}}(n+k,n;\gamma)}
  =
  \frac{K^{\ast}_{n+k,n}}{K^{\mathrm{fix}}_{n+k,n}}
  \leq \frac{k}{n}<1.
\end{equation*}
At
the left endpoint \(k=2\), a matching converse proves
\(K^{\ast}_{n+2,n}=8\), thereby showing that
\(\Pe^{\ast}(n+2,n;\gamma)\) is asymptotically smaller by the factor
\(n/2\) than \(\Pe^{\mathrm{fix}}(n+2,n;\gamma)\).  At the terminal
endpoint \(k=n\), universal optimality of
the full cross-polytope proves
\(K^{\ast}_{2n,n}=K^{\mathrm{fix}}_{2n,n}=4n(n-1)\).  For
\(3\leq k\leq n-1\), we conjecture that
\(K^{\ast}_{n+k,n}=4k(k-1)\).
\end{enumerate}

\subsection{Relation to Existing Theory}
\label{subsec:related-work}

For a fixed finite codebook over the AWGN channel, the pairwise error is
a Gaussian tail, and the closest-pair expansion is asymptotically
equivalent to the exact ML message error probability
\cite{Proakis2001,AlvaradoAgrellBrannstrom2018}.  Because these results
hold the codebook fixed, they do not justify minimizing the resulting
formula over a different codebook at every SNR.  At the exponential
scale, equal-energy codebook design is the classical spherical-packing
problem
\cite{Shannon1959,Rankin1955,ConwaySloane1999,EricsonZinoviev2001}, for
which linear- and semidefinite-programming methods provide powerful
general bounds
\cite{DelsarteGoethalsSeidel1977,Levenshtein1998,
BachocVallentin2008,BachocVallentin2009}.

The simplex application uses the simplex bound \cite{Rankin1955} and
the universal optimality of embedded regular simplices
\cite{CohnKumar2007}.  The orthoplex-bound application uses the orthoplex
bound \cite{Rankin1955} and its equality-case classification
\cite[Thm.~3]{Kuperberg2007}; an equivalent characterization is given in
\cite[Thm.~2.2]{KingMixonParshallWells2026}.  It also uses the universal
optimality of the full cross-polytope \cite{CohnKumar2007}.  The recent result in
\cite[Cors.~2.5--2.6]{Mulgund2026WeakSimplex} proves the stronger finite-SNR
coding statement that, for \(M\) equiprobable equal-energy signals over
the real AWGN channel, an embedded regular \((M-1)\)-simplex maximizes
the average probability of correct ML decoding at every positive SNR
whenever \(n\geq M-1\).
That result also characterizes the equality case: the embedded regular
simplex is the unique maximizer up to an orthogonal transformation and
relabeling.

A closely related large-parameter phenomenon occurs for Riesz energies.
In \cite{BondarenkoHardinSaff2014}, the authors proved that
cluster points of fixed-cardinality
Riesz-energy minimizers as the exponent tends to infinity are best
packings.  For five points on \(\Sph^2\), they identified, up to isometry,
the unique limiting square-pyramid configuration and proved that the
normalized minimum Riesz energy converges to \(8\); their upper bound
uses an exponent-dependent square-pyramid deformation.
This is a close antecedent to the \((M,n)=(5,3)\) specialization of
the SNR-dependent construction in our work.
Their Riesz objective, however, neither gives the Gaussian construction
for general \((n+k,n)\) nor transfers its asymptotics to exact AWGN ML error.

In a different positive-temperature cross-entropy problem motivated by
neural collapse, \cite{AlcalaEtAl2026NeuralCollapse} show
that, among packing-optimal
codes in the orthoplex regime, sufficiently low temperature selects a
``low-entropy'' geometry consisting of one regular simplex and an
orthogonal cross-polytope.
For \(M=n+k\), this is the same simplex-plus-antipodal-pairs geometry that
attains the best fixed-codebook leading coefficient here.  Their feasible class
and objective are different, and their result does not address
SNR-dependent paths outside the packing-optimal set.

The preceding works provide geometric, coding, and energy results that
serve as technical inputs or conceptual antecedents for particular
simplex and orthoplex-bound applications considered here.  The theory
developed in this paper is broader: it is not restricted to these
special geometries and characterizes the optimal high-SNR leading coefficient
when a finite spherical codebook may be reoptimized at each SNR.  Its
new ingredients are the uniform transfer from the soft-packing energy
to exact ML error, the existence and characterization of the optimal
leading coefficient, and the localization and closest-pair contribution
results for competitive SNR-dependent families.  The simplex and
orthoplex-bound results are concrete applications of this general framework.

\subsection{Organization}
\label{subsec:organization}

Section~\ref{sec:signal-packing-fixed} formulates the exact design problem
and its fixed-codebook benchmark.
Section~\ref{sec:soft-packing-characterization} develops the soft-packing
characterization of the optimal leading coefficient and establishes
attainment by near-minimizers of the Gaussian soft-packing energy.
Section~\ref{sec:geometric-reduction} establishes localization, the
SNR-dependent closest-pair structure of competitive families, and the
global variational characterization by jointly feasible scaled offsets.
Section~\ref{sec:optimization-attainment-extraction} develops the
matching-bounds criterion used in the applications.
Sections~\ref{sec:simplex-regime} and \ref{sec:orthoplex-regime} apply
this theory to the simplex and orthoplex-bound ranges, respectively.
Section~\ref{sec:conclusion} concludes the paper, and the
appendices provide the technical proofs.

\section{Signal Model and High-SNR Codebook Design Problem}
\label{sec:signal-packing-fixed}

This section introduces the exact AWGN error probability, the
spherical-packing quantities that determine its optimal exponential
decay, and the classical high-SNR expansion for a fixed codebook.  It
then defines the exact-error benchmark obtained by restricting the design
to packing-optimal codebooks and formulates the unrestricted optimization
performed separately at each SNR.

\subsection{Signal Model and Exact Error Probability}
\label{subsec:awgn-model}

Fix integers \(M\ge2\) and \(n\ge1\), write
\([M]=\{1,\ldots,M\}\), and let
\(\Sph^{n-1}=\{\bx\in\R^n:\|\bx\|=1\}\).
A labeled spherical codebook is
\begin{equation*}
  \cC=(\bc_1,\ldots,\bc_M)\in\cM,
  \qquad
  \cM=(\Sph^{n-1})^M,
\end{equation*}
where every codeword has unit norm.  The product space \(\cM\) is
compact.

Let \(W\) be uniform on \([M]\).  The normalized real AWGN channel is
\begin{equation}
  \bY=\sqrt{n\gamma}\,\bc_W+\bZ,
  \qquad
  \bZ\sim\Normal(\bzero,\bI_n).
  \label{eq:awgn-channel}
\end{equation}
The noise variance is one per real channel use, and the transmitted
vector has energy \(n\gamma\); hence \(\gamma\) is the average SNR per
real channel use.  Unless stated otherwise, every high-SNR limit keeps
\((M,n)\) fixed and lets \(\gamma\to\infty\).

For an observation \(\by\in\R^n\), the maximum-likelihood decoder is
\begin{equation}
  \widehat W(\by)\in\argmax_{1\le m\le M}\by^\T\bc_m.
  \label{eq:ml-decoder}
\end{equation}
Because all codewords have the same norm, maximizing correlation is
equivalent to minimizing Euclidean distance.  We use an arbitrary fixed
measurable rule to break ties.  With equiprobable messages, the average
error is independent of that choice, and ties have probability zero when
no two codewords coincide.

The exact average message error probability is
\begin{equation}
  \Pe(\cC;\gamma)
  =\frac1M\sum_{m=1}^M
  \Pp\{\widehat W(\bY)\ne m\mid W=m\}.
  \label{eq:Pe-def}
\end{equation}

\subsection{Pairwise Geometry and the Packing Benchmark}
\label{subsec:pairwise-packing}

Let
\begin{equation*}
  \cE_M=\{(m,\ell):m,\ell\in[M],\ m\ne\ell\}
\end{equation*}
be the ordered-pair index set.  For \(e=(m,\ell)\in\cE_M\), define
the pairwise correlation
\begin{equation}
  \rho_e(\cC)=\bc_m^\T\bc_\ell.
  \label{eq:pairwise-correlation-def}
\end{equation}
Because the codewords have unit norm,
\begin{equation}
  \|\bc_m-\bc_\ell\|^2=2\bigl(1-\rho_e(\cC)\bigr).
  \label{eq:correlation-distance-relation}
\end{equation}
Thus a larger correlation means a smaller Euclidean distance, and
minimizing the largest pairwise correlation is equivalent to solving the
spherical-code, or spherical-cap-packing, problem
\cite{Rankin1955,ConwaySloane1999,EricsonZinoviev2001}.

The largest correlation of a particular codebook and the smallest such
largest correlation attainable over all codebooks are, respectively,
\begin{equation}
  \rho_{\max}(\cC)=\max_{e\in\cE_M}\rho_e(\cC),
  \qquad
  \rho^{\ast}_{M,n}=\min_{\cC\in\cM}\rho_{\max}(\cC).
  \label{eq:rho-max-rho-star}
\end{equation}
The minimum exists because \(\cM\) is compact and \(\rho_{\max}\) is
continuous.
Thus \(\rho_{\max}(\cC)\) is the worst correlation within the particular
codebook \(\cC\): it identifies its closest pair or pairs.  By contrast,
\(\rho^{\ast}_{M,n}\) is the best worst correlation achievable among all \((M,n)\)
spherical codebooks.  Equivalently,
\(2(1-\rho^{\ast}_{M,n})\) is the largest achievable minimum squared
distance.

The packing-optimal codebook set is
\begin{equation}
  \cP^{\ast}
  =\{\cC\in\cM:\rho_{\max}(\cC)=\rho^{\ast}_{M,n}\}.
  \label{eq:Pstar-def}
\end{equation}
Hence, \(\cP^{\ast}\) is the collection of all codebooks whose worst
correlation equals the best worst correlation \(\rho^{\ast}_{M,n}\).
It contains every solution of the spherical-packing problem, including
rotated and relabeled copies and, when they exist, geometrically
noncongruent optimal packings.

The analysis in this paper concerns spherical codes for which
\(\rho^{\ast}_{M,n}<1\).  Then every packing-optimal codebook has
pairwise-distinct codewords, and its optimal minimum squared distance is
\(2(1-\rho^{\ast}_{M,n})>0\).  The value
\(\rho^{\ast}_{M,n}\) determines the best high-SNR error exponent.

An unordered codeword pair \(\{m,\ell\}\) is a \emph{closest pair} of
\(\cC\) if it attains the minimum intercodeword distance, or equivalently
if \(\rho_{(m,\ell)}(\cC)=\rho_{\max}(\cC)\).  Collect the corresponding
ordered closest-pair indices in
\begin{equation}
  \cI_{\max}(\cC)
  =\{e\in\cE_M:\rho_e(\cC)=\rho_{\max}(\cC)\}.
  \label{eq:Imax-def}
\end{equation}
For each unordered closest pair, this set contains both directed indices
\((m,\ell)\) and \((\ell,m)\).  If
\(\cC\in\cP^{\ast}\), every index in \(\cI_{\max}(\cC)\) has
correlation \(\rho^{\ast}_{M,n}\).

\subsection{Fixed-Codebook High-SNR Asymptotics}
\label{subsec:fixed-codebook-asymptotics}

Assume a codebook \(\cC\) is \emph{pairwise distinct}, i.e.,
\(\rho_{\max}(\cC)<1\) holds.  For such a codebook and distinct \(m,\ell\),
conditioned on \(W=m\), let
\begin{equation*}
  \cA_{m\ell}
  =
  \{(\sqrt{n\gamma}\,\bc_m+\bZ)^\T\bc_\ell
  \ge(\sqrt{n\gamma}\,\bc_m+\bZ)^\T\bc_m\}
\end{equation*}
be the event that competitor \(\ell\) scores at least as highly as the
transmitted codeword.  For \(e=(m,\ell)\), write
\(\cA_e:=\cA_{m\ell}\).  The standard pairwise Gaussian calculation gives
\begin{equation}
  \Pp(\cA_{m\ell})
  =
  \Qfun\!\left(
  \sqrt{\frac{n\gamma[1-\rho_{(m,\ell)}(\cC)]}{2}}
  \right),
  \label{eq:pairwise-tail}
\end{equation}
where
\(\Qfun(x)=(2\pi)^{-1/2}\int_x^\infty\exp\{-t^2/2\}\dd t\)
\cite{Proakis2001}.  Consequently,
\begin{equation}
  \Pe(\cC;\gamma)
  =\frac1M\sum_{m=1}^M
  \Pp\!\left(\bigcup_{\ell\ne m}\cA_{m\ell}\right).
  \label{eq:exact-error-event-union}
\end{equation}

For \(\rho<1\), define the single-ordered-pair tail scale
\begin{equation}
  b_\gamma(\rho)
  =
  \frac{\exp\{-n\gamma(1-\rho)/4\}}
  {M\sqrt{n\gamma}\sqrt{\pi(1-\rho)}}.
  \label{eq:pairwise-tail-scale}
\end{equation}
The standard fixed-codebook high-SNR result then takes the following
form.

\begin{proposition}[Fixed-codebook high-SNR expansion]
\label{prop:fixed-codebook-asymptotics}
Let \(\cC\) be fixed with pairwise-distinct codewords, equivalently
\(\rho_{\max}(\cC)<1\).  Then
\begin{equation}
\begin{aligned}
  \Pe(\cC;\gamma)
  &=\frac{|\cI_{\max}(\cC)|}{M}
  \Qfun\!\left(
  \sqrt{\frac{n\gamma[1-\rho_{\max}(\cC)]}{2}}
  \right)\bigl(1+o_{\cC}(1)\bigr)\\
  &=b_\gamma(\rho_{\max}(\cC))
  \bigl(|\cI_{\max}(\cC)|+o_{\cC}(1)\bigr).
\end{aligned}
  \label{eq:fixed-codebook-leading-expansion}
\end{equation}
Consequently,
\begin{equation}
  \lim_{\gamma\to\infty}-\gamma^{-1}\log\Pe(\cC;\gamma)
  =\frac{n[1-\rho_{\max}(\cC)]}{4}.
  \label{eq:fixed-codebook-error-exponent}
\end{equation}
\end{proposition}

Here and throughout, \(\log\) denotes the natural logarithm.

The proposition follows by specializing the standard fixed-constellation
result \cite[Th.~3]{AlvaradoAgrellBrannstrom2018} to equal-prior
spherical codebooks.
It is the conventional nearest-neighbor principle stated directly as an
expansion of the exact error: after normalization by the tail scale
associated with its own worst correlation, a fixed codebook has leading
coefficient \(|\cI_{\max}(\cC)|\).
The remainder \(o_{\cC}(1)\) is pointwise in \(\cC\): it tends to zero
for each fixed codebook, but its rate may vary with the codebook.

Evaluating \eqref{eq:pairwise-tail-scale} at the optimal packing
correlation defines the single-ordered-pair packing-tail scale
\begin{equation}
  B_\gamma^{\ast}
  :=b_\gamma(\rho^{\ast}_{M,n})
  =\frac{\exp\{-n\gamma(1-\rho^{\ast}_{M,n})/4\}}
  {M\sqrt{n\gamma}\sqrt{\pi(1-\rho^{\ast}_{M,n})}}.
  \label{eq:Bstar}
\end{equation}
It contains the optimal exponential decay, the Gaussian-tail prefactor,
and the \(1/M\) message average, but no closest-pair multiplicity.  Any fixed
codebook outside \(\cP^{\ast}\) has
\(\rho_{\max}(\cC)>\rho^{\ast}_{M,n}\), and its exact error divided by
\(B_\gamma^{\ast}\) therefore grows exponentially.  Hence only
packing-optimal codebooks can serve as fixed-codebook benchmarks at the
optimal exponent, and every fixed \(\cC_0\in\cP^{\ast}\) satisfies
\begin{equation}
  \frac{\Pe(\cC_0;\gamma)}{B_\gamma^{\ast}}
  \longrightarrow|\cI_{\max}(\cC_0)|.
  \label{eq:fixed-packing-normalized-limit}
\end{equation}
The best fixed-codebook leading coefficient attainable among
packing-optimal codebooks is therefore
\begin{equation}
  K^{\mathrm{fix}}_{M,n}
  :=\min_{\cC_0\in\cP^{\ast}}|\cI_{\max}(\cC_0)|.
  \label{eq:fixed-packing-coefficient}
\end{equation}
Equivalently, among all packing-optimal codebooks, one chooses a
codebook with the fewest ordered closest-pair indices and uses that same
codebook as the SNR increases.  Moreover,
\(K^{\mathrm{fix}}_{M,n}/M\) is the
smallest average kissing number among packing-optimal codebooks.

To compare exact error probabilities, define the best packing-optimal
benchmark at SNR \(\gamma\) by
\begin{equation}
  \Pe^{\mathrm{fix}}(M,n;\gamma)
  :=\min_{\cC_0\in\cP^{\ast}}\Pe(\cC_0;\gamma).
  \label{eq:Pe-fix-def}
\end{equation}
The minimum is attained because \(\cP^{\ast}\) is compact and the exact
average ML error is continuous on this set.  We show later, cf.
Corollary~\ref{cor:packing-optimal-benchmark}, that
\begin{equation}
  \Pe^{\mathrm{fix}}(M,n;\gamma)
  =B_\gamma^{\ast}
  \bigl(K^{\mathrm{fix}}_{M,n}+o(1)\bigr).
  \label{eq:Pe-fix-asymptotic}
\end{equation}
Accordingly, \(K^{\mathrm{fix}}_{M,n}\) and
\(\Pe^{\mathrm{fix}}(M,n;\gamma)\) are the best fixed-codebook
leading-coefficient and exact-error benchmarks against which the SNR-wise
optimized quantities are compared.

\subsection{Exact High-SNR Codebook Design Problem}
\label{subsec:exact-design-problem}

For fixed \((M,n)\) and operating SNR \(\gamma\), the exact
SNR-wise optimal average error probability is
\begin{equation}
  \Pe^{\ast}(M,n;\gamma)
  =\min_{\cC\in\cM}\Pe(\cC;\gamma).
  \label{eq:Pe-star-def}
\end{equation}
Thus the minimization ranges over all labeled \(M\)-tuples of unit
vectors in \(\R^n\), and
\(\Pe^{\ast}(M,n;\gamma)\) denotes the smallest exact average ML error
probability among them.  The minimum is attained because \(\cM\) is
compact and the average ML error is continuous in the codeword tuple.

An optimizing codebook may be denoted by
\begin{equation*}
  \cC_\gamma^{\ast}
  \in\argmin_{\cC\in\cM}\Pe(\cC;\gamma).
\end{equation*}
The subscript emphasizes that the optimizing codebook may change with
SNR; there is no requirement that
\(\cC_{\gamma_1}^{\ast}=\cC_{\gamma_2}^{\ast}\).

Because \(\cP^{\ast}\subseteq\cM\), the two exact design problems in
\eqref{eq:Pe-fix-def} and \eqref{eq:Pe-star-def} obey
\begin{equation}
  \Pe^{\ast}(M,n;\gamma)
  \leq\Pe^{\mathrm{fix}}(M,n;\gamma)
  \qquad\text{for every \(\gamma\)}.
  \label{eq:Pe-star-le-Pe-fix}
\end{equation}

Every fixed \(\cC_0\in\cP^{\ast}\) achieves the best error exponent.
Among such fixed codebooks, however, only a minimizer in
\eqref{eq:fixed-packing-coefficient} is guaranteed to have the smallest
fixed-codebook leading coefficient.  Determining the unrestricted exact
finite-SNR minimum \(\Pe^{\ast}(M,n;\gamma)\) and its optimal leading coefficient
requires optimization over the full codebook manifold.

The central problem of this paper is therefore to determine whether the limit
\begin{equation}
  \lim_{\gamma\to\infty}
  \frac{\Pe^{\ast}(M,n;\gamma)}{B_\gamma^{\ast}}
  \label{eq:central-coefficient-problem}
\end{equation}
exists and is finite and, if so, to characterize it.  Once existence is
proved, the limit is
denoted by \(K^{\ast}_{M,n}\) and is called the \emph{optimal leading
coefficient}.  
For any SNR-dependent family, if
\(\Pe(\cC_\gamma;\gamma)/B_\gamma^{\ast}\) converges to a finite limit,
we call that limit the \emph{leading coefficient of the family}.  For a
general SNR-dependent family on which no asymptotic optimality condition
is imposed, this finite limit need not exist: the normalized error may be
unbounded, or it may remain bounded while approaching different finite
values along different SNR subsequences.  This possible nonexistence
concerns the leading coefficient of such a family and does not arise for
the optimal leading coefficient.  As shown later, cf.
Theorem~\ref{thm:global-coefficient},
\(\Pe^{\ast}(M,n;\gamma)/B_\gamma^{\ast}\) converges to a finite limit
and hence the optimal leading coefficient \(K^{\ast}_{M,n}\) exists.

  An
SNR-dependent family \(\{\cC_\gamma\}_\gamma\) attains the optimal leading
coefficient if it satisfies
\begin{equation}
\begin{aligned}
  \frac{\Pe(\cC_\gamma;\gamma)}{B_\gamma^{\ast}}
  &\longrightarrow K^{\ast}_{M,n},\\
  \frac{\Pe(\cC_\gamma;\gamma)}
  {\Pe^{\ast}(M,n;\gamma)}
  &\longrightarrow1.
\end{aligned}
  \label{eq:central-problem-attainment}
\end{equation}

A second objective is to investigate whether
\(K^{\ast}_{M,n}<K^{\mathrm{fix}}_{M,n}\) can occur and to characterize
the associated geometric mechanism.  A strict inequality means
that SNR-dependent motion of the codewords toward the packing-optimal set improves the
optimal leading coefficient even though both designs have the same optimal
error exponent.  Together with \eqref{eq:Pe-fix-asymptotic}, the main
optimized expansion proved below then gives
\begin{equation}
  \frac{\Pe^{\ast}(M,n;\gamma)}
       {\Pe^{\mathrm{fix}}(M,n;\gamma)}
  \longrightarrow
  \frac{K^{\ast}_{M,n}}{K^{\mathrm{fix}}_{M,n}}.
  \label{eq:Pe-star-Pe-fix-ratio}
\end{equation}
Hence \(K^{\ast}_{M,n}<K^{\mathrm{fix}}_{M,n}\) implies
\(\Pe^{\ast}(M,n;\gamma)<\Pe^{\mathrm{fix}}(M,n;\gamma)\) for all
sufficiently large \(\gamma\).

\subsection{How \texorpdfstring{\(K^{\ast}_{M,n}<K^{\mathrm{fix}}_{M,n}\)}
{K* < Kfix} Can Occur}
\label{subsec:strict-coefficient-inequality}

The best fixed-codebook leading coefficient \(K^{\mathrm{fix}}_{M,n}\)
and the optimal leading coefficient \(K^{\ast}_{M,n}\) arise from
different high-SNR procedures.
The coefficient \(K^{\mathrm{fix}}_{M,n}\) is obtained by holding a
packing-optimal codebook fixed as the SNR grows, identifying its number
of ordered closest-pair indices, and then minimizing this number over
\(\cP^{\ast}\).  By contrast, \(K^{\ast}_{M,n}\) is attained
asymptotically by an SNR-dependent family
\(\{\cC_\gamma\}_\gamma\).  By definition, no codebook in
\(\cP^{\ast}\) has fewer than \(K^{\mathrm{fix}}_{M,n}\) ordered
closest-pair indices.  Hence, an attaining
family under the strict inequality
\(K^{\ast}_{M,n}<K^{\mathrm{fix}}_{M,n}\), cf.
Lemma~\ref{lem:coefficient-attaining-localization}, approaches \(\cP^{\ast}\)
while remaining outside it for all sufficiently large \(\gamma\).
Hence, the strict inequality arises from the relative rates at which an
SNR-dependent codebook family approaches the packing-optimal set
\(\cP^{\ast}\).

To make this more precise, along a
convergent subsequence, reindexed by \(\gamma\), write
\(\cC_\gamma\to\cC_\infty\in\cP^{\ast}\).  For an ordered pair
\(e\in\cI_{\max}(\cC_\infty)\), define its signed correlation offset
from the optimal packing value and compare its single-pair tail scale
with \(B_\gamma^{\ast}\):
\begin{align}
  \Delta_{e,\gamma}
  &:=\rho_e(\cC_\gamma)-\rho^{\ast}_{M,n},
  \label{eq:section2-contact-gap}\\
  \frac{b_\gamma(\rho_e(\cC_\gamma))}{B_\gamma^{\ast}}
  &=
  \exp\!\left\{\frac{n\gamma\Delta_{e,\gamma}}4\right\}
  \sqrt{\frac{1-\rho^{\ast}_{M,n}}
  {1-\rho_e(\cC_\gamma)}}.
  \label{eq:pathwise-pair-scale-ratio}
\end{align}
Because \(\rho_e(\cC_\gamma)\to\rho^{\ast}_{M,n}\), the square-root
factor tends to one.  As the limiting-closest-pair result below makes
rigorous, for the family under discussion that attains the optimal leading
coefficient, if
\(n\gamma\Delta_{e,\gamma}\to
\lambda_e\in\mathbb{R}\cup\{-\infty\}\), then the contribution of
the ordered pair \(e\) to \(K^{\ast}_{M,n}\) is
\(\exp\{\lambda_e/4\}\), where \(\exp\{-\infty\}=0\).  Thus a zero
scaled offset gives a unit contribution to \(K^{\ast}_{M,n}\), whereas divergence to
\(-\infty\) suppresses the contribution even though the pair is closest
in the packing-optimal limit \(\cC_\infty\).

A revealing mechanism for achieving
\(K^{\ast}_{M,n}<K^{\mathrm{fix}}_{M,n}\), applied in the
orthoplex-bound range, is the following.  Partition the closest-pair indices
in \(\cI_{\max}(\cC_\infty)\) into two classes, denoting indices in the
selected class by \(e\) and those in the remaining class by \(f\).  At
finite SNR \(\gamma\), arrange their correlations so that
\begin{equation}
\begin{aligned}
  \rho_e(\cC_\gamma)
  &=\rho_{\max}(\cC_\gamma)
    =\rho^{\ast}_{M,n}+p_\gamma,
  & p_\gamma&>0,
  & p_{\gamma}&\longrightarrow0,
  \qquad n\gamma p_\gamma\longrightarrow0,\\
  \rho_f(\cC_\gamma)
  &=\rho^{\ast}_{M,n}-q_{f,\gamma},
  & q_{f,\gamma}&>0,
  & q_{f,\gamma}&\longrightarrow0,
  \qquad n\gamma q_{f,\gamma}\longrightarrow\infty.
\end{aligned}
  \label{eq:section2-two-scale-mechanism}
\end{equation}
The first pairs are slightly closer than under an optimal packing, but
their deterioration is too small to change the optimal error exponent,
and they make unit contributions to the family's leading coefficient.
The second pairs are farther apart by a larger vanishing amount than
under an optimal packing, so their contributions to the family's leading
coefficient vanish, even though they become closest pairs in
\(\cC_\infty\).  Thus an asymptotically negligible worsening of selected
pairs can create enough geometric freedom to suppress the contributions
of other pairs to the family's leading coefficient.  This description
assumes that the family's leading coefficient exists.  More general
approach rates may produce fractional pair contributions, so the leading
coefficient of an SNR-dependent family, when it exists, need not be an
integer-valued ordered closest-pair count.  If the family attains the
optimal leading coefficient, its leading coefficient equals
\(K^{\ast}_{M,n}\).

If \(\cC_\infty\) were instead held fixed, every one of its ordered
closest-pair indices would contribute one to its fixed-codebook leading
coefficient \(\lvert\cI_{\max}(\cC_\infty)\rvert\), which is at least
\(K^{\mathrm{fix}}_{M,n}\).  For an SNR-dependent family that attains the
optimal leading coefficient, the pairwise contributions to
\(K^{\ast}_{M,n}\) therefore depend on the family's approach to the
limiting packing, not on the limiting packing alone.

\section{Soft-Packing Characterization of the Optimal Leading Coefficient}
\label{sec:soft-packing-characterization}

The design problem in Subsection~\ref{subsec:exact-design-problem} first
asks whether the normalized SNR-wise optimal ML error has a finite limit
and how that limit can be characterized.  This section resolves these
questions for every fixed finite pair \((M,n)\) satisfying
\(M\ge2\) and \(\rho^{\ast}_{M,n}<1\).  We introduce a Gaussian
soft-packing energy, prove that its minimum is asymptotically equivalent
to the normalized exact ML optimum, and show that this minimum converges
to the optimal leading coefficient \(K^{\ast}_{M,n}\).  Comparison with fixed
packing-optimal codebooks then gives
\(K^{\ast}_{M,n}\le K^{\mathrm{fix}}_{M,n}\).

\subsection{Gaussian Soft-Packing Energy}
\label{subsec:global-soft-packing-energy}

For every codebook \(\cC\in\cM\), define the Gaussian soft-packing energy as
\begin{equation}
  \cG_\gamma(\cC)
  =\sum_{e\in\cE_M}
  \exp\left\{\frac{n\gamma}{4}
  \bigl(\rho_e(\cC)-\rho^{\ast}_{M,n}\bigr)\right\}.
  \label{eq:global-soft-packing-energy}
\end{equation}
Each exponential term in \(\cG_\gamma(\cC)\) distinguishes correlation offsets on the
\(1/(n\gamma)\) scale: an offset \(\lambda/(n\gamma)\) produces the
finite term \(\exp\{\lambda/4\}\) in \(\cG_\gamma\).
The Gaussian soft-packing energy also controls the hard packing objective.  Indeed, an
ordered pair attaining \(\rho_{\max}(\cC)\) contributes
\begin{equation*}
  \exp\left\{\frac{n\gamma}{4}
  \bigl(\rho_{\max}(\cC)-\rho^{\ast}_{M,n}\bigr)\right\}
\end{equation*}
to \(\cG_\gamma(\cC)\).  Therefore, if
\(\cG_\gamma(\cC)\le L\) for some \(L\geq1\), then
\begin{equation}
  0\le \rho_{\max}(\cC)-\rho^{\ast}_{M,n}
  \le \frac{4\log L}{n\gamma}.
  \label{eq:bounded-energy-packing-consequence}
\end{equation}
Thus a bounded Gaussian soft-packing energy \(\cG_\gamma(\cC)\) forces the worst correlation into the
\(1/(n\gamma)\)-scale neighborhood of the optimal packing value \(\rho^{\ast}_{M,n}\).

For a fixed packing-optimal codebook \(\cC_0\in\cP^{\ast}\), every
\(e\in\cI_{\max}(\cC_0)\) has zero correlation offset from \(\rho^{\ast}_{M,n}\) and thereby contributes
one to \(\cG_\gamma(\cC_0)\).  Every remaining ordered pair has a fixed
negative offset from \(\rho^{\ast}_{M,n}\); hence its term decays
exponentially to zero as \(\gamma\to\infty\).  Consequently,
\begin{equation}
  \lim_{\gamma\to\infty}\cG_\gamma(\cC_0)
  =|\cI_{\max}(\cC_0)|.
  \label{eq:soft-packing-fixed-limit}
\end{equation}
Thus the limiting Gaussian soft-packing energy of a fixed packing-optimal codebook is
exactly its number of ordered closest-pair indices.
We next compare the Gaussian soft-packing energy with the normalized exact ML
error while the codebook is allowed to vary with SNR.

\subsection{Uniform Approximation of the Exact ML Error}
\label{subsec:uniform-soft-packing-transfer}

The minimum of the Gaussian soft-packing energy over all codebooks
\(\cC\in\cM\) at a fixed SNR \(\gamma\) is
\begin{equation}
  V_\gamma(M,n)
  =\min_{\cC\in\cM}\cG_\gamma(\cC).
  \label{eq:soft-packing-optimum}
\end{equation}
For each fixed \(\gamma\), the map
\(\cC\mapsto\cG_\gamma(\cC)\) is continuous on the compact codebook
space \(\cM\); hence the minimum is attained.  The quantity
\(V_\gamma(M,n)\) plays a crucial role in characterizing the optimal leading
coefficient.

\begin{theorem}[Uniform Gaussian soft-packing transfer]
\label{thm:uniform-soft-packing-transfer}
Fix finite \((M,n)\) with \(M\ge2\) and \(\rho^{\ast}_{M,n}<1\).  For every fixed \(L\ge1\), there are constants \(C_L<\infty\) and \(\gamma_L<\infty\) such that, for every \(\gamma\ge\gamma_L\) and every \(\cC\in\cM\) satisfying
\begin{equation}
  \cG_\gamma(\cC)\le L,
  \label{eq:bounded-soft-packing-class}
\end{equation}
we have
\begin{equation}
  \left|
  \frac{\Pe(\cC;\gamma)}{B_\gamma^{\ast}}
  -\cG_\gamma(\cC)
  \right|
  \le \frac{C_L}{n\gamma}.
  \label{eq:uniform-soft-packing-transfer}
\end{equation}
The constants may depend on \((M,n)\), \(L\), and the positive
separation \(1-\rho^{\ast}_{M,n}\), but not on \(\cC\) or \(\gamma\).
\end{theorem}
\begin{IEEEproof}
The proof is given in
Appendix~\ref{app:global-optimization-proofs}.
\end{IEEEproof}

Theorem~\ref{thm:uniform-soft-packing-transfer} shows that, for every fixed \(L\geq1\),
the soft-packing energy \(\cG_\gamma(\cC)\) uniformly approximates the
normalized exact ML error
\(\frac{\Pe(\cC;\gamma)}{B_\gamma^{\ast}}\)
over all codebooks \(\cC\in\cM\) satisfying
\(\cG_\gamma(\cC)\leq L\).  Equivalently, an SNR-dependent family is
packing-tail competitive precisely when
\(\Pe(\cC_\gamma;\gamma)/B_\gamma^{\ast}=O(1)\).
Later on, cf. Lemma~\ref{lem:coefficient-attaining-localization}, we show that every
packing-tail-competitive family eventually satisfies such a fixed
soft-packing energy bound.

\begin{corollary}[Best packing-optimal benchmark]
\label{cor:packing-optimal-benchmark}
For every fixed finite \((M,n)\) with \(M\ge2\) and
\(\rho^{\ast}_{M,n}<1\),
\begin{equation}
  \frac{\Pe^{\mathrm{fix}}(M,n;\gamma)}{B_\gamma^{\ast}}
  =K^{\mathrm{fix}}_{M,n}+O((n\gamma)^{-1}).
  \label{eq:Pe-fix-uniform-expansion}
\end{equation}
\end{corollary}

\begin{IEEEproof}
For every \(\cC_0\in\cP^{\ast}\), all correlation offsets from
\(\rho^{\ast}_{M,n}\) are nonpositive, and hence
\(\cG_\gamma(\cC_0)\le M(M-1)\).  Applying
Theorem~\ref{thm:uniform-soft-packing-transfer} uniformly over
\(\cP^{\ast}\) and then taking minima gives
\begin{equation*}
  \frac{\Pe^{\mathrm{fix}}(M,n;\gamma)}{B_\gamma^{\ast}}
  =\min_{\cC_0\in\cP^{\ast}}\cG_\gamma(\cC_0)
  +O((n\gamma)^{-1}).
\end{equation*}
Every \(\cC_0\in\cP^{\ast}\) satisfies
\(\cG_\gamma(\cC_0)\ge|\cI_{\max}(\cC_0)|
\ge K^{\mathrm{fix}}_{M,n}\).  Choose
\(\cC_0^{\mathrm{fix}}\in\cP^{\ast}\) with
\(|\cI_{\max}(\cC_0^{\mathrm{fix}})|=K^{\mathrm{fix}}_{M,n}\).
Its remaining ordered pairs, if any, lie a fixed positive correlation
gap below \(\rho^{\ast}_{M,n}\).  Therefore,
\begin{equation*}
  \cG_\gamma(\cC_0^{\mathrm{fix}})
  =K^{\mathrm{fix}}_{M,n}+O(\exp\{-c n\gamma\})
\end{equation*}
for some \(c>0\).  These two bounds prove
\eqref{eq:Pe-fix-uniform-expansion}.
\end{IEEEproof}

\begin{corollary}[Finite-SNR equivalence of the optimized exact ML error and Gaussian soft-packing energy]
\label{cor:optimized-soft-packing-transfer}
For every fixed finite \((M,n)\) with \(M\ge2\) and \(\rho^{\ast}_{M,n}<1\),
\begin{equation}
  \frac{\Pe^{\ast}(M,n;\gamma)}{B_\gamma^{\ast}}
  =V_\gamma(M,n)+O((n\gamma)^{-1}).
  \label{eq:optimized-soft-packing-transfer}
\end{equation}
\end{corollary}

\begin{IEEEproof}
The proof is given in
Appendix~\ref{app:global-optimization-proofs}.
\end{IEEEproof}

The corollary shows that
minimizing the normalized exact ML error and minimizing the soft-packing energy, see \eqref{eq:soft-packing-optimum},
produce asymptotically identical optimal normalized values.  This result
compares the two optimized values,
\(\Pe^{\ast}(M,n;\gamma)/B_\gamma^{\ast}\) and \(V_\gamma(M,n)\), at each
sufficiently large finite SNR.
It does not yet establish the existence of an
optimal leading coefficient, because it does not show that
\(V_\gamma(M,n)\) converges as \(\gamma\to\infty\).

\subsection{Existence and Characterization of the Optimal Leading Coefficient}
\label{subsec:global-coefficient-attainment}

It remains to prove that \(V_\gamma(M,n)\) has a high-SNR limit.
The next theorem establishes this convergence and transfers the
resulting limit to the minimum error probability using
Corollary~\ref{cor:optimized-soft-packing-transfer}.

\begin{theorem}[Existence and Gaussian soft-packing energy characterization of the optimal leading coefficient]
\label{thm:global-coefficient}
For every fixed finite \((M,n)\) with \(M\ge2\) and
\(\rho^{\ast}_{M,n}<1\), the finite limits
\begin{equation}
\begin{aligned}
  K^{\ast}_{M,n}
  &:=\lim_{\gamma\to\infty}V_\gamma(M,n)\\
  &=\lim_{\gamma\to\infty}
  \frac{\Pe^{\ast}(M,n;\gamma)}{B_\gamma^{\ast}}
\end{aligned}
  \label{eq:Kstar-soft-packing-limit}
\end{equation}
exist and agree.  The common value, \(K^{\ast}_{M,n}\), satisfies
\begin{equation}
  2\le K^{\ast}_{M,n}\le M(M-1).
  \label{eq:Kstar-elementary-bounds}
\end{equation}
The globally optimal ML error obeys
\begin{equation}
  \Pe^{\ast}(M,n;\gamma)
  =K^{\ast}_{M,n}B_\gamma^{\ast}+o(B_\gamma^{\ast}).
  \label{eq:global-main}
\end{equation}
Moreover, for every fixed packing-optimal codebook \(\cC_0\in\cP^{\ast}\),
\begin{equation}
  K^{\ast}_{M,n}\le |\cI_{\max}(\cC_0)|.
  \label{eq:Kstar-fixed-packing-upper}
\end{equation}
\end{theorem}

\begin{IEEEproof}
The proof is given in
Appendix~\ref{app:global-optimization-proofs}.
\end{IEEEproof}

The theorem shows that both
\(\Pe^{\ast}(M,n;\gamma)/B_\gamma^{\ast}\) and \(V_\gamma(M,n)\)
converge to \(K^{\ast}_{M,n}\), and that this common limit is at most
\(\lvert\cI_{\max}(\cC_0)\rvert\) for every fixed packing-optimal
codebook \(\cC_0\).  The scale \(B_\gamma^{\ast}\) contains the
packing-determined exponential decay and its Gaussian prefactor, while
\(K^{\ast}_{M,n}\) is the remaining optimal leading coefficient.

The theorem also proves existence of the optimized
leading coefficient. Specifically, both
\(\Pe^{\ast}(M,n;\gamma)/B_\gamma^{\ast}\) and \(V_\gamma(M,n)\)
converge to \(K^{\ast}_{M,n}\).  Consequently, every selection of exact
finite-SNR minimizers
\[
  \cC_\gamma^{\ast}
  \in\argmin_{\cC\in\cM}\Pe(\cC;\gamma)
\]
has leading coefficient \(K^{\ast}_{M,n}\), even if the selected
optimizing codebooks do not themselves converge as a family.  Indeed,
\[
  \frac{\Pe(\cC_\gamma^{\ast};\gamma)}{B_\gamma^{\ast}}
  =
  \frac{\Pe^{\ast}(M,n;\gamma)}{B_\gamma^{\ast}}
  \longrightarrow K^{\ast}_{M,n}.
\]

The same characterization also identifies a broad class of
families attaining the optimal leading coefficient.  Exact minimization of the Gaussian
soft-packing energy at every SNR is unnecessary; a globally vanishing
additive objective gap is sufficient, as shown in the following.

\begin{corollary}[Near-minimizers of the Gaussian soft-packing energy attain the optimal leading coefficient]
\label{cor:soft-energy-minimizers-attain}
Let \(\{\widehat{\cC}_\gamma\}\subset\cM\) and suppose that, for some
\(\eps_\gamma\ge0\) satisfying \(\eps_\gamma\to0\),
\begin{equation}
  0\le
  \cG_\gamma(\widehat{\cC}_\gamma)-V_\gamma(M,n)
  \le\eps_\gamma.
  \label{eq:soft-energy-near-minimizer-gap}
\end{equation}
Then
\begin{equation}
\begin{aligned}
  \cG_\gamma(\widehat{\cC}_\gamma)
  &\longrightarrow K^{\ast}_{M,n},\\
  \frac{\Pe(\widehat{\cC}_\gamma;\gamma)}{B_\gamma^{\ast}}
  &\longrightarrow K^{\ast}_{M,n},
\end{aligned}
  \label{eq:soft-energy-minimizers-attain-Kstar}
\end{equation}
and
\begin{equation}
  \frac{\Pe(\widehat{\cC}_\gamma;\gamma)}
  {\Pe^{\ast}(M,n;\gamma)}\longrightarrow1.
  \label{eq:soft-energy-near-minimizer-relative-attainment}
\end{equation}
\end{corollary}

\begin{IEEEproof}
The proof is given in
Appendix~\ref{app:global-optimization-proofs}.
\end{IEEEproof}

Thus, any full SNR-indexed family
\(\{\widehat{\cC}_\gamma\}_\gamma\) whose Gaussian soft-packing energy
\(\cG_\gamma(\widehat{\cC}_\gamma)\) is within an additive \(o(1)\) of
the global minimum \(V_\gamma(M,n)\) as \(\gamma\to\infty\) attains the
optimal leading coefficient \(K^{\ast}_{M,n}\), even if its codebooks do
not minimize the exact ML error at any finite SNR.

\subsection{Comparison With the Best Packing-Optimal Benchmark}
The fixed-codebook comparison can be written directly via the Gaussian soft-packing energy function as
\begin{equation}
\begin{aligned}
  K^{\ast}_{M,n}
  &=
  \lim_{\gamma\to\infty}
  \min_{\cC\in\cM}\cG_\gamma(\cC)\\
  &\le
  \min_{\cC_0\in\cP^{\ast}}
  \lim_{\gamma\to\infty}\cG_\gamma(\cC_0)
  =K^{\mathrm{fix}}_{M,n}.
\end{aligned}
\label{eq:noncommuting-energy-limits}
\end{equation}
On the left, the codebook is optimized separately at every SNR before
the high-SNR limit is taken.  On the right, a packing-optimal codebook
is first held fixed, its limiting number of ordered closest-pair indices is obtained, and
that count is then minimized over \(\cP^{\ast}\).  The first procedure
retains the path-dependent pair contributions to the optimal leading
coefficient that can arise under SNR-wise codebook optimization; the second
assigns one unit to the fixed codebook's leading coefficient for every
ordered closest-pair index of that codebook.

In fact, \eqref{eq:Pe-star-le-Pe-fix} gives
\begin{equation*}
  \Pe^{\ast}(M,n;\gamma)
  \leq\Pe^{\mathrm{fix}}(M,n;\gamma)
  \qquad\text{for every \(\gamma\)}.
\end{equation*}
Theorem~\ref{thm:global-coefficient} and
Corollary~\ref{cor:packing-optimal-benchmark} then give
\begin{equation}
\begin{aligned}
  \frac{\Pe^{\ast}(M,n;\gamma)}
       {\Pe^{\mathrm{fix}}(M,n;\gamma)}
  &\longrightarrow
  \frac{K^{\ast}_{M,n}}{K^{\mathrm{fix}}_{M,n}},\\
  \frac{\Pe^{\mathrm{fix}}(M,n;\gamma)
        -\Pe^{\ast}(M,n;\gamma)}
       {B_\gamma^{\ast}}
  &\longrightarrow
  K^{\mathrm{fix}}_{M,n}-K^{\ast}_{M,n}.
\end{aligned}
\label{eq:optimized-exact-vs-fixed-packing}
\end{equation}
Therefore, under the strict coefficient inequality,
\(K^{\ast}_{M,n}<K^{\mathrm{fix}}_{M,n}\), the packing-optimal
restriction retains a constant
multiplicative disadvantage, and the exact SNR-wise optimal error is strictly
smaller for every sufficiently large SNR.  If the coefficients are equal,
\(K^{\ast}_{M,n}=K^{\mathrm{fix}}_{M,n}\),
the best fixed packing-optimal codebook is leading-order optimal, although
lower-order differences may remain.

The limit characterization determines the optimal leading coefficient
but does not yet describe the geometry of the codebook
families that can approach it.  The next section supplies that structural
description and derives a global variational characterization.

\section{Localization and SNR-Dependent Closest-Pair Contributions}
\label{sec:geometric-reduction}

The preceding section identifies the optimal leading coefficient as the
high-SNR limit of the minimum Gaussian soft-packing energy,
\(K^{\ast}_{M,n}=\lim_{\gamma\to\infty}V_\gamma(M,n)\).
This section determines which ordered pairs can have nonvanishing normalized
contributions along a family competitive at the packing-tail scale and how
those contributions are determined.     The section concludes by minimizing these
contributions over all packing-optimal limiting codebooks and all jointly
feasible scaled closest-pair offsets.

\subsection{Localization at the Packing-Tail Scale}
\label{subsec:coefficient-attaining-localization}

Here and below, \(\operatorname{dist}(\cC,\cP^{\ast})\) is induced by
the product Euclidean metric on the labeled tuples in \(\cM\).

\begin{lemma}[Localization at the packing-tail scale]
\label{lem:coefficient-attaining-localization}
Let \(\{\cC_\gamma\}\subset\cM\) satisfy
\begin{equation}
  \frac{\Pe(\cC_\gamma;\gamma)}{B_\gamma^{\ast}}=O(1).
  \label{eq:packing-tail-competitive-family}
\end{equation}
Then
\begin{equation}
  \rho_{\max}(\cC_\gamma)-\rho^{\ast}_{M,n}
  =O((n\gamma)^{-1}),
  \label{eq:attaining-family-packing-gap}
\end{equation}
\begin{equation}
  \cG_\gamma(\cC_\gamma)=O(1),
  \qquad
  \operatorname{dist}(\cC_\gamma,\cP^{\ast})\longrightarrow0.
  \label{eq:attaining-family-Pstar-localization}
\end{equation}
If the family additionally attains the optimal leading coefficient in the
sense of \eqref{eq:central-problem-attainment} and
\(K^{\ast}_{M,n}<K^{\mathrm{fix}}_{M,n}\), then
\begin{equation}
  \cC_\gamma\notin\cP^{\ast}
  \quad\text{for all sufficiently large \(\gamma\)}.
  \label{eq:attaining-family-eventually-outside-Pstar}
\end{equation}
\end{lemma}

\begin{IEEEproof}
The proof is given in
Appendix~\ref{app:local-reduction-proofs}.
\end{IEEEproof}

Lemma~\ref{lem:coefficient-attaining-localization} shows that any
packing-tail-competitive codebook family must enter an increasingly
small neighborhood of the packing-optimal set.
Indeed, define the excess maximal correlation as
\(\delta_\gamma=
\rho_{\max}(\cC_\gamma)-\rho^{\ast}_{M,n}\).
If \(\delta_\gamma>0\), then the single-pair Gaussian-tail scale,
normalized by \(B_\gamma^{\ast}\), grows essentially as
\(
\exp\!\left\{\frac{n\gamma\delta_\gamma}{4}\right\}.
\)
Consequently, bounded normalized error is possible only if
\eqref{eq:attaining-family-packing-gap} holds.
Thus, the maximal correlation must be within
\(O((n\gamma)^{-1})\) of \(\rho^{\ast}_{M,n}\), which is the critical
correlation scale governing pairwise contributions to a family's leading
coefficient, whenever that coefficient exists.
The lemma also establishes the geometric localization
\(\operatorname{dist}(\cC_\gamma,\cP^{\ast})\longrightarrow0\).
Hence, the family may continue
to rotate, relabel its codewords, or move among different packing-optimal
configurations while its distance from \(\cP^{\ast}\) tends to
zero as the SNR grows.

The lemma also reveals the geometric mechanism by
which the optimal leading coefficient can be smaller than the best
fixed-codebook leading coefficient.  If
\(K^{\ast}_{M,n}<K^{\mathrm{fix}}_{M,n}\), then a family attaining the
optimal leading coefficient must lie outside \(\cP^{\ast}\) for every
sufficiently large \(\gamma\).  Instead, it approaches
\(\cP^{\ast}\) from outside.  Such a family accepts a vanishing increase
in its maximal correlation, confined to the \(1/(n\gamma)\) scale, in order
to reduce or suppress the contributions to \(K^{\ast}_{M,n}\) made by other
pairs that are closest in a packing-optimal subsequential limit.  Therefore, such a
family may consist of codebooks that are slightly worse as packings at
every sufficiently large finite SNR and nevertheless have a strictly smaller leading
coefficient than every fixed packing-optimal codebook.

\subsection{Reduction to Closest Pairs of a Limiting Packing}
\label{subsec:packing-contact-reduction}

Consider a convergent subsequence, reindexed by \(\gamma\), for which
\(\cC_\gamma\to\cC_\infty\in\cP^{\ast}\).  Thus
\(\cC_\infty\) is a packing-optimal subsequential limit of the family.  For each
\(e\in\cI_{\max}(\cC_\infty)\), use the signed correlation offset from
\eqref{eq:section2-contact-gap},
\begin{equation*}
  \Delta_{e,\gamma}
  =\rho_e(\cC_\gamma)-\rho^{\ast}_{M,n}.
\end{equation*}
A positive offset worsens that ordered pair relative to the packing
value, whereas a negative offset improves it.

\begin{theorem}[Reduction to closest pairs of a limiting packing]
\label{thm:packing-contact-reduction}
Let \(\cC_\gamma\to\cC_\infty\in\cP^{\ast}\), and suppose that, for
some finite \(L\geq1\),
\(\cG_\gamma(\cC_\gamma)\le L\) for all sufficiently large \(\gamma\).
Then, for all sufficiently large \(\gamma\),
\begin{equation*}
  0\le
  \rho_{\max}(\cC_\gamma)-\rho^{\ast}_{M,n}
  =\max_{e\in\cI_{\max}(\cC_\infty)}\Delta_{e,\gamma}
  \le\frac{4\log L}{n\gamma}.
\end{equation*}
Moreover, for some \(c>0\),
\begin{equation}
  \cG_\gamma(\cC_\gamma)
  =
  \sum_{e\in\cI_{\max}(\cC_\infty)}
  \exp\left\{\frac{n\gamma\Delta_{e,\gamma}}4\right\}
  +O(\exp\{-c n\gamma\}),
  \label{eq:packing-contact-energy-reduction}
\end{equation}
\begin{equation}
\begin{split}
  \frac{\Pe(\cC_\gamma;\gamma)}{B_\gamma^{\ast}}
  &={}
  \sum_{e\in\cI_{\max}(\cC_\infty)}
  \exp\left\{\frac{n\gamma\Delta_{e,\gamma}}4\right\}\\
  &\quad+O((n\gamma)^{-1}).
\end{split}
  \label{eq:packing-contact-error-reduction}
\end{equation}
\end{theorem}

\begin{IEEEproof}
The proof is given in
Appendix~\ref{app:local-reduction-proofs}.
\end{IEEEproof}

Theorem~\ref{thm:packing-contact-reduction} shows that, for a
packing-tail-competitive codebook family that converges to an optimal
packing, not all ordered codeword pairs can contribute to the normalized
error \(\Pe(\cC_\gamma;\gamma)/B_\gamma^{\ast}\).  The limiting packing
restricts the ordered pairs that can contribute to the normalized error to
those that are closest pairs in the limiting optimal packing.  Indeed, if an
ordered pair \(e\) is not in
\(\cI_{\max}(\cC_\infty)\),
then, for some fixed \(\eta_e>0\),
\(\rho_e(\cC_\infty)=\rho^{\ast}_{M,n}-\eta_e\).
Because \(\cC_\gamma\to\cC_\infty\), this pair retains a
strict correlation gap for all sufficiently large \(\gamma\).  Therefore,
its Gaussian soft-packing energy term is bounded by an exponentially vanishing quantity
of the form \(\exp\{-n\gamma\eta_e/8\}\), and its contribution to the normalized error vanishes.
For an ordered pair that is closest in the limiting packing, i.e., \(e \in
\cI_{\max}(\cC_\infty)\), by contrast,
\(\rho_e(\cC_\infty)=\rho^{\ast}_{M,n}\),
so there is no fixed correlation gap that makes its contribution to the
normalized error vanish.  Instead, its finite-SNR displacement
\(\Delta_{e,\gamma}=
\rho_e(\cC_\gamma)-\rho^{\ast}_{M,n}\)
enters the normalized error through the term
\(\exp\{n\gamma\Delta_{e,\gamma}/4\}\).
Consequently, the theorem reduces the normalized error to
\eqref{eq:packing-contact-error-reduction}.

Ordinary geometric convergence determines the set of potentially relevant
closest pairs but does not determine each contribution in this sum.  It
gives only \(\Delta_{e,\gamma}\to0\), whereas the limiting normalized error
along a subsequence depends on the finer quantities
\(n\gamma\Delta_{e,\gamma}\).  Thus, different SNR-dependent approaches to
the same optimal packing can yield different subsequential limits of the
normalized error.  The limiting codebook is the common destination,
whereas the approach to that destination on the \(1/(n\gamma)\) correlation
scale determines each limiting closest pair's normalized contribution along
the selected subsequence.  The following corollary makes these contributions
explicit when the scaled offsets have limits.

\begin{corollary}[Subsequential closest-pair contribution formula]
\label{cor:subsequential-weighted-contact}
Under the assumptions of
Theorem~\ref{thm:packing-contact-reduction}, since
\(\cI_{\max}(\cC_\infty)\) is finite and a bounded Gaussian soft-packing energy
precludes any scaled offset from diverging to \(+\infty\), there exists
a further subsequence, without changing notation, such that
\begin{equation}
  n\gamma\Delta_{e,\gamma}\longrightarrow\lambda_e,
  \qquad
  e\in\cI_{\max}(\cC_\infty),
  \label{eq:weighted-contact-lambda-limit}
\end{equation}
simultaneously for some
\(\lambda_e\in\R\cup\{-\infty\}\), where
\(\exp\{-\infty\}=0\).  Along this subsequence,
\begin{align}
  \cG_\gamma(\cC_\gamma)
  &\longrightarrow
  \sum_{e\in\cI_{\max}(\cC_\infty)}
  \exp\{\lambda_e/4\},
  \label{eq:weighted-contact-energy-limit}\\
  \frac{\Pe(\cC_\gamma;\gamma)}{B_\gamma^{\ast}}
  &\longrightarrow
  \sum_{e\in\cI_{\max}(\cC_\infty)}
  \exp\{\lambda_e/4\}.
  \label{eq:weighted-contact-error-limit}
\end{align}

If the family attains the optimal leading coefficient in the sense of
\eqref{eq:central-problem-attainment},
then
\[
  K^{\ast}_{M,n}
  =
  \sum_{e\in\cI_{\max}(\cC_\infty)}
  \exp\{\lambda_e/4\}.
\]

For the constant family
\(\cC_\gamma\equiv\cC_\infty\), every
\(e\in\cI_{\max}(\cC_\infty)\) has \(\Delta_{e,\gamma}=0\) and
therefore contributes one to both limiting sums.  Hence,
\[
  \lim_{\gamma\to\infty}\cG_\gamma(\cC_\infty)
  =
  \lim_{\gamma\to\infty}
  \frac{\Pe(\cC_\infty;\gamma)}{B_\gamma^{\ast}}
  =
  \bigl|\cI_{\max}(\cC_\infty)\bigr|
  \geq K^{\mathrm{fix}}_{M,n}.
\]
Thus, a fixed packing-optimal codebook has a leading coefficient equal to
its number of ordered closest-pair indices; this coefficient equals
\(K^{\ast}_{M,n}\) precisely when that constant family attains the optimal
leading coefficient.
\end{corollary}

\begin{IEEEproof}
Equations~\eqref{eq:packing-contact-energy-reduction} and
\eqref{eq:packing-contact-error-reduction} give the two displayed
limits.  Bounded Gaussian soft-packing energy gives
\(n\gamma\Delta_{e,\gamma}\le4\log L\) for every
\(e\in\cI_{\max}(\cC_\infty)\).
The extended interval \([{-\infty},4\log L]\) is compact, and the
set \(\cI_{\max}(\cC_\infty)\) is finite, so a common convergent
subsequence exists.  The attainment conclusion follows from
\eqref{eq:central-problem-attainment}.  The constant-family conclusion
follows from \(\Delta_{e,\gamma}\equiv0\) and
\eqref{eq:fixed-packing-coefficient}.
\end{IEEEproof}

Corollary~\ref{cor:subsequential-weighted-contact} converts the scaled
correlation offset of each limiting closest pair into its contribution
\(\exp\{\lambda_e/4\}\) to the limiting normalized error along that
subsequence.  If the full family has a leading coefficient, these
contributions sum to that coefficient.  If the family attains the optimal
leading coefficient, their sum is \(K^{\ast}_{M,n}\).

A fixed packing-optimal codebook provides the simplest special case.  If
\(\cC_\gamma=\cC_\infty\), then every ordered closest-pair index has
zero offset, so that \(\lambda_e=0\) and \(\exp\{\lambda_e/4\}=1\).
Its leading coefficient is therefore the number of ordered closest-pair
indices.  For an SNR-dependent family, however, ordinary convergence
still makes every such offset tend to zero, but multiplication by
\(n\gamma\) reveals differences that remain visible after normalization
by \(B_\gamma^{\ast}\).  In particular, if
\(n\gamma\Delta_{e,\gamma}\longrightarrow\lambda_e\), the limiting
pair contribution has one of the following forms:
\[
\begin{array}{c|c}
\text{scaled-offset behavior} & \text{limiting pair contribution}\\
\hline
\lambda_e=-\infty
    & \exp\{\lambda_e/4\}=0,\\
-\infty<\lambda_e<0
    & 0<\exp\{\lambda_e/4\}<1,\\
\lambda_e=0
    & \exp\{\lambda_e/4\}=1,\\
\lambda_e>0
    & \exp\{\lambda_e/4\}>1.
\end{array}
\]
For a family attaining the optimal leading coefficient, these are the corresponding
contributions to \(K^{\ast}_{M,n}\).
Consequently, a pair that is closest in the limiting codebook may have a
zero, fractional, unit, or greater-than-unit contribution, depending on the
\(1/(n\gamma)\)-scale path by which the codebook approaches the limiting
packing.

\subsection{Global Variational Characterization by Scaled Closest-Pair Offsets}
\label{subsec:weighted-contact-variational-characterization}

Corollary~\ref{cor:subsequential-weighted-contact} evaluates one specified
SNR-dependent family.  To characterize the global optimum, the feasible
subsequential scaled offsets must be considered jointly over every
packing-optimal limit.

For \(\cC_0\in\cP^{\ast}\), let
\(\mathfrak L(\cC_0)\) denote the set of all vectors
\begin{equation*}
  \bm\lambda
  =(\lambda_e)_{e\in\cI_{\max}(\cC_0)}
  \in
  \bigl(\R\cup\{-\infty\}\bigr)^{\cI_{\max}(\cC_0)}
\end{equation*}
for which there exist SNRs \(\gamma_j\to\infty\) and codebooks
\(\cC_j\in\cM\) satisfying
\begin{equation}
\begin{aligned}
  \cC_j&\longrightarrow\cC_0,
  &\sup_j\cG_{\gamma_j}(\cC_j)&<\infty,\\
  n\gamma_j\bigl(\rho_e(\cC_j)-\rho^{\ast}_{M,n}\bigr)
  &\longrightarrow\lambda_e,
  &e&\in\cI_{\max}(\cC_0).
\end{aligned}
\label{eq:admissible-contact-limit-set}
\end{equation}
We call \(\bm\lambda\) admissible scaled closest-pair-offset data at
\(\cC_0\).  The bounded Gaussian soft-packing energy condition excludes \(+\infty\) as a
component limit, whereas \(-\infty\) is retained and gives the limiting
pair weight \(\exp\{-\infty\}=0\).  The constant sequence
\(\cC_j\equiv\cC_0\) shows that \(\mathfrak L(\cC_0)\) is nonempty.

\begin{proposition}[Global variational characterization by scaled closest-pair offsets]
\label{prop:weighted-contact-variational-characterization}
For every fixed finite \((M,n)\) with \(M\ge2\) and
\(\rho^{\ast}_{M,n}<1\),
\begin{equation}
  K^{\ast}_{M,n}
  =
  \min_{\substack{
    \cC_0\in\cP^{\ast}\\
    \bm\lambda\in\mathfrak L(\cC_0)}}
  \sum_{e\in\cI_{\max}(\cC_0)}
  \exp\{\lambda_e/4\}.
  \label{eq:weighted-contact-variational-characterization}
\end{equation}
The minimum in \eqref{eq:weighted-contact-variational-characterization}
is attained.  Moreover, every
sequence witnessing a minimizing pair \((\cC_0,\bm\lambda)\) satisfies,
along its SNR sequence,
\begin{equation*}
  \cG_{\gamma_j}(\cC_j)\longrightarrow K^{\ast}_{M,n},
  \qquad
  \frac{\Pe(\cC_j;\gamma_j)}{B_{\gamma_j}^{\ast}}
  \longrightarrow K^{\ast}_{M,n},
  \qquad
  \frac{\Pe(\cC_j;\gamma_j)}
  {\Pe^{\ast}(M,n;\gamma_j)}\longrightarrow1.
\end{equation*}
\end{proposition}

\begin{IEEEproof}
The proof is given in
Appendix~\ref{app:local-reduction-proofs}.
\end{IEEEproof}

The proposition minimizes over every packing-optimal limiting codebook and
every jointly feasible set of scaled offsets.  The limiting packing fixes
the ordered pairs in the sum, while the offsets determine their limiting
weights.  A constant packing has \(\lambda_e=0\) at every ordered
closest-pair index and therefore recovers the ordinary closest-pair count.
The joint-feasibility condition is essential: the scaled offsets must
arise simultaneously from one feasible codebook sequence and therefore
cannot be optimized independently.  The characterization of
\(K^{\ast}_{M,n}\) is exact but generally implicit because the admissible
sets \(\mathfrak L(\cC_0)\) need not admit tractable finite-dimensional
descriptions.

\section{Leading-Coefficient Extraction by Matching Bounds}
\label{sec:optimization-attainment-extraction}

The preceding sections establish the existence of \(K^{\ast}_{M,n}\)
and characterize the closest-pair contributions that can produce it.  To
determine \(K^{\ast}_{M,n}\) in closed form, the applications below use a
direct matching-bounds argument.  Specifically, since
\begin{equation*}
  K^{\ast}_{M,n}=\lim_{\gamma\to\infty}V_\gamma(M,n),
\end{equation*}
it is enough to identify a candidate value \(K\), prove the global lower
bound \(\liminf_{\gamma\to\infty}V_\gamma(M,n)\geq K\), and construct a
full SNR-indexed family whose Gaussian soft-packing energy has limiting
upper bound at most \(K\).  The following corollary then identifies \(K\)
as the optimal leading coefficient, transfers the equality to the exact
SNR-wise optimized ML error, and certifies the constructed family.

\begin{corollary}[Optimal leading-coefficient extraction by matching bounds]
\label{cor:coefficient-extraction}
Fix a finite pair \((M,n)\) with \(M\ge2\) and
\(\rho^{\ast}_{M,n}<1\), and let \(K\in[2,M(M-1)]\).  Suppose that
\begin{equation}
  \liminf_{\gamma\to\infty}V_\gamma(M,n)\ge K,
  \label{eq:coefficient-extraction-lower}
\end{equation}
and that there is an SNR-indexed family, defined for every sufficiently
large \(\gamma\),
\(\{\cC_\gamma\}\subset\cM\) satisfying
\begin{equation}
  \limsup_{\gamma\to\infty}
  \cG_\gamma(\cC_\gamma)\le K.
  \label{eq:coefficient-extraction-upper}
\end{equation}
Then
\begin{equation}
  K^{\ast}_{M,n}=K,
  \qquad
  \Pe^{\ast}(M,n;\gamma)
  =B_\gamma^{\ast}\bigl(K+o(1)\bigr),
  \label{eq:coefficient-extraction-conclusion}
\end{equation}
and the constructed family satisfies
\begin{equation}
  \Pe(\cC_\gamma;\gamma)
  =B_\gamma^{\ast}\bigl(K+o(1)\bigr),
  \qquad
  \frac{\Pe(\cC_\gamma;\gamma)}
  {\Pe^{\ast}(M,n;\gamma)}\to1.
  \label{eq:coefficient-extraction-attainment}
\end{equation}
\end{corollary}

\begin{IEEEproof}
Since \(V_\gamma(M,n)\le\cG_\gamma(\cC_\gamma)\),
\begin{align*}
  K
  &\le\liminf_{\gamma\to\infty}V_\gamma(M,n)\\
  &\le\limsup_{\gamma\to\infty}V_\gamma(M,n)\\
  &\le\limsup_{\gamma\to\infty}\cG_\gamma(\cC_\gamma)\\
  &\le K.
\end{align*}
Thus \(V_\gamma(M,n)\to K\), and
Theorem~\ref{thm:global-coefficient} gives
\(K^{\ast}_{M,n}=K\) and the optimal leading-order expansion.  Moreover,
\(V_\gamma(M,n)\le\cG_\gamma(\cC_\gamma)\) gives
\(\liminf_{\gamma\to\infty}\cG_\gamma(\cC_\gamma)\ge K\), while
\eqref{eq:coefficient-extraction-upper} gives the reverse upper limit.
Hence \(\cG_\gamma(\cC_\gamma)\to K\).  This convergence places the
family in a fixed bounded-energy sublevel of \(\cG_\gamma\), so
Theorem~\ref{thm:uniform-soft-packing-transfer} yields
\(\Pe(\cC_\gamma;\gamma)/B_\gamma^{\ast}=K+o(1)\).  Together with the
optimal leading-order expansion and \(K\ge2\), this also gives the relative-error
limit in \eqref{eq:coefficient-extraction-attainment}.
\end{IEEEproof}

The lower bound in \eqref{eq:coefficient-extraction-lower} is the
converse: it controls every SNR-dependent codebook family through the
global minimum \(V_\gamma(M,n)\).  The upper bound in
\eqref{eq:coefficient-extraction-upper} is achievability: one feasible
family reaches the same value.  Once the two bounds match,
Corollary~\ref{cor:coefficient-extraction} and the uniform transfer in
Theorem~\ref{thm:uniform-soft-packing-transfer} give both the exact
SNR-wise optimized ML expansion and attainment of the optimal leading
coefficient.   

Accordingly, determining the optimal leading coefficient \(K^{\ast}_{M,n}\) by matching
bounds consists of three steps:
\begin{enumerate}
\item Determine the optimal packing correlation \(\rho^{\ast}_{M,n}\),
which fixes the universal tail scale \(B_\gamma^{\ast}\), and identify a
candidate value \(K\).

\item Prove the global converse
\begin{equation*}
  \liminf_{\gamma\to\infty}V_\gamma(M,n)\ge K.
\end{equation*}
This may follow from a direct inequality valid for every feasible
codebook or from localization followed by a problem-specific geometric
argument.

\item Construct a feasible codebook family, defined for every
sufficiently large \(\gamma\), such that
\begin{equation*}
  \limsup_{\gamma\to\infty}
  \cG_\gamma(\cC_\gamma)\le K.
\end{equation*}
Corollary~\ref{cor:coefficient-extraction} then proves
\(K^{\ast}_{M,n}=K\), the exact SNR-wise optimized ML expansion, and relative-error
attainment of the constructed family.
\end{enumerate}

The next two sections apply this three-step matching-bounds framework to
the simplex and orthoplex-bound ranges.

\section{Application to the Simplex Codes: Exact Optimal Leading Coefficients}
\label{sec:simplex-regime}

This section applies the matching-bounds framework to
\(2\leq M\leq n+1\).  The simplex bound in \cite{Rankin1955} fixes the packing scale,
and a global Jensen inequality solves the soft-packing energy optimization at
every positive SNR.  The resulting optimal leading coefficient is determined exactly, is attained by
a fixed regular simplex, and cannot be improved by SNR-dependent motion.

Fix integers \(n\ge1\) and \(2\leq M\leq n+1\), and let
\(\cC_{\mathrm{simp}}\) be an embedded regular \((M-1)\)-simplex in
\(\Sph^{n-1}\).  This bound gives
\begin{equation}
  \rho^{\ast}_{M,n}=-\frac1{M-1},
  \label{eq:simplex-packing-correlation}
\end{equation}
with equality only for an embedded regular simplex, up to an orthogonal
transformation and relabeling \cite{Rankin1955}.  Substitution in
\eqref{eq:Bstar} fixes the packing-tail scale as
\begin{equation}
  B_\gamma^{\ast}
  =\frac{\exp\{-n\gamma M/[4(M-1)]\}}
  {M\sqrt{\pi n\gamma M/(M-1)}}.
  \label{eq:simplex-packing-scale}
\end{equation}
Here the finite-SNR Gaussian soft-packing energy optimization can be solved exactly.  The
following proposition gives the global minimum and identifies all
minimizers.

\begin{proposition}[Exact soft-packing energy minimum in the simplex range]
\label{prop:simplex-energy-minimum}
For all integers \(n\geq1\) and \(2\leq M\leq n+1\), and every
\(\gamma>0\),
\begin{equation}
  V_\gamma(M,n)
  =\cG_\gamma(\cC_{\mathrm{simp}})
  =M(M-1).
  \label{eq:simplex-energy-minimum}
\end{equation}
Moreover, the minimizers are exactly the embedded regular
\((M-1)\)-simplices, up to orthogonal transformations and relabeling.
\end{proposition}
\begin{IEEEproof}
The proof is given in
Appendix~\ref{app:simplex-orthoplex-application-proofs}.
\end{IEEEproof}

The proposition
is stronger than a high-SNR statement: an embedded regular simplex
globally minimizes the Gaussian soft-packing energy for every \(\gamma>0\).    Consequently, no finite-SNR deformation,
including one that varies with \(\gamma\), can lower the Gaussian soft-packing energy
minimum below the simplex value.

\begin{corollary}[Optimal leading coefficient in the simplex range]
\label{cor:simplex-coefficient}
For every fixed pair \((M,n)\) satisfying \(2\leq M\leq n+1\),
\begin{equation}
  K^{\ast}_{M,n}
  =K^{\mathrm{fix}}_{M,n}
  =M(M-1),
  \label{eq:simplex-optimized-coefficient}
\end{equation}
and
\begin{equation}
\begin{split}
  \Pe^{\ast}(M,n;\gamma)
  &=B_\gamma^{\ast}\bigl(M(M-1)+o(1)\bigr)\\
  &=\frac{(M-1)^{3/2}}
  {\sqrt{\pi n\gamma M}}
  \exp\left\{-\frac{n\gamma M}{4(M-1)}\right\}
  \bigl(1+o(1)\bigr).
\end{split}
  \label{eq:simplex-exact-error-expansion}
\end{equation}
The constant family \(\cC_\gamma\equiv\cC_{\mathrm{simp}}\) attains
\(K^{\ast}_{M,n}\) and satisfies
\begin{equation}
  \frac{\Pe(\cC_{\mathrm{simp}};\gamma)}
  {\Pe^{\ast}(M,n;\gamma)}
  \longrightarrow1.
  \label{eq:simplex-coefficient}
\end{equation}
\end{corollary}

\begin{IEEEproof}
Proposition~\ref{prop:simplex-energy-minimum} gives the global lower
bound \eqref{eq:coefficient-extraction-lower} and, through the constant
simplex family, the matching upper bound
\eqref{eq:coefficient-extraction-upper}, both with
\(K=M(M-1)\).  Corollary~\ref{cor:coefficient-extraction} gives
\(K^{\ast}_{M,n}=M(M-1)\) and attainment of the optimal leading coefficient;
substitution of \eqref{eq:simplex-packing-scale} gives
\eqref{eq:simplex-exact-error-expansion}.  Packing uniqueness and the
fact that every ordered pair of distinct simplex vertices is a
closest-pair index give
\(K^{\mathrm{fix}}_{M,n}=M(M-1)\), completing
\eqref{eq:simplex-optimized-coefficient}.  The relative-error assertion
then follows from attainment of the optimal leading coefficient.
\end{IEEEproof}

Corollary~\ref{cor:simplex-coefficient} completely resolves the
SNR-wise optimized high-SNR problem in the simplex range.  The optimal
and best fixed-codebook leading coefficients coincide, and the constant
simplex family attains the optimal leading coefficient; SNR-dependent motion therefore gives no
leading-order gain.  Proposition~\ref{prop:simplex-energy-minimum}
establishes exact finite-SNR optimality for the Gaussian soft-packing energy, while the
conclusion for the exact ML error follows asymptotically through uniform
transfer.  Independently, the recent resolution of the Weak
Simplex Conjecture in \cite[Cors.~2.5--2.6]{Mulgund2026WeakSimplex} proves the stronger statement that an embedded regular
\((M-1)\)-simplex minimizes the exact average ML error at every positive
SNR in this entire cardinality range.
That external result is consistent with, but is not needed for, the
leading-coefficient derivation above.

\section{Application to the Orthoplex-Bound Range}
\label{sec:orthoplex-regime}

Fix \(n\ge2\) and parameterize the orthoplex-bound range by
\begin{equation}
  M=n+k,
  \qquad 2\leq k\leq n,
  \label{eq:orthoplex-range-parameterization}
\end{equation}
so that \(n+2\leq M\leq2n\).  This section first determines the packing
scale and the best fixed-codebook leading coefficient
\(K^{\mathrm{fix}}_{n+k,n}\).
It then constructs a strictly improving SNR-dependent family for every
\(2\leq k<n\), obtains an achievability bound \(K^{\ast}_{n+k,n}\leq 4k(k-1)\), determines \(K^{\ast}_{n+k,n}\) exactly at
\(k=2\) and \(k=n\), and isolates the remaining intermediate equality
as a conjecture.

\subsection{Packing Scale and Best Fixed-Codebook Leading Coefficient}
\label{subsec:orthoplex-fixed-benchmark}

The orthoplex bound in \cite{Rankin1955} gives
\(\rho^{\ast}_{M,n}\ge0\) whenever \(M>n+1\).
Conversely, since \(M=n+k\leq2n\), any \(M\)
vertices of the cross-polytope
\(\{\pm\be_1,\ldots,\pm\be_n\}\) have all off-diagonal correlations
in \(\{-1,0\}\).  Hence, throughout the range,
\begin{equation}
  \rho^{\ast}_{n+k,n}=0.
  \label{eq:orthoplex-range-packing-correlation}
\end{equation}
Substitution in \eqref{eq:Bstar} gives the packing-tail scale
\begin{equation}
  B_\gamma^{\ast}
  =
  \frac{\exp\{-n\gamma/4\}}
  {(n+k)\sqrt{n\gamma}\sqrt{\pi}}.
  \label{eq:orthoplex-range-packing-scale}
\end{equation}

\begin{proposition}[Best fixed-codebook leading coefficient in the orthoplex-bound range]
\label{prop:orthoplex-range-fixed-contact}
For every fixed \(n\ge2\) and \(2\leq k\leq n\),
\begin{equation}
  K^{\mathrm{fix}}_{n+k,n}=4n(k-1).
  \label{eq:orthoplex-range-fixed-contact-count}
\end{equation}
Equality is attained by the orthogonal union of \(k-1\) antipodal pairs
and a regular \((n-k+1)\)-simplex, the latter having \(n-k+2\)
vertices in the remaining \((n-k+1)\)-dimensional subspace.
\end{proposition}
\begin{IEEEproof}
The proof is given in
Appendix~\ref{app:simplex-orthoplex-application-proofs}.
\end{IEEEproof}

Proposition~\ref{prop:orthoplex-range-fixed-contact} determines the
smallest leading coefficient obtainable by holding one packing-optimal
codebook fixed.  The stated orthogonal union attains the minimum ordered
closest-pair-index count \(4n(k-1)\).  This is the benchmark for measuring the
benefit of SNR-wise codebook optimization, developed below.

\subsection{An SNR-Dependent Construction for
\texorpdfstring{\(2\leq k<n\)}{2 <= k < n}}
\label{subsec:orthoplex-range-upper-bound}

For every fixed \(n\ge3\) and \(2\leq k<n\), put
\begin{equation}
  d=n-k,
  \qquad
  q_\gamma=\frac{\log(n\gamma)}{n\gamma}.
  \label{eq:orthoplex-range-moving-parameters}
\end{equation}
The closest-pair contribution formula in
Corollary~\ref{cor:subsequential-weighted-contact} suggests this choice
because it separates the two scales
\[
  \frac{1}{n\gamma}\ll q_\gamma
  \ll\frac{1}{\sqrt{n\gamma}}.
\]
Thus correlations lowered by order \(q_\gamma\) produce zero limiting
weights in the Gaussian soft-packing energy, whereas correlations worsened
by order \(q_\gamma^2\) produce unit limiting weights.

For all sufficiently large \(\gamma\), choose unit vectors
\(\ba_1,\ldots,\ba_d\) in a \(d\)-dimensional subspace
\(\mathcal U\subset\R^n\), continuously in \(q_\gamma\) from an
orthonormal basis at \(q_\gamma=0\), such that
\begin{equation}
  \ba_i^\T\ba_j=-q_\gamma,
  \qquad i\ne j.
  \label{eq:orthoplex-range-auxiliary-gram}
\end{equation}
Define
\begin{equation}
\begin{aligned}
  \bv_\gamma
  &=-\frac{q_\gamma}{1-(d-1)q_\gamma}
  \sum_{j=1}^d\ba_j,\\
  p_\gamma
  &=\|\bv_\gamma\|^2
  =\frac{dq_\gamma^2}{1-(d-1)q_\gamma},\\
  r_\gamma
  &=\sqrt{1-p_\gamma}.
\end{aligned}
  \label{eq:orthoplex-range-displacement}
\end{equation}
Let \(\bm f_1,\ldots,\bm f_k\) be a fixed orthonormal basis of
\(\mathcal U^\perp\), and set
\begin{equation}
  \cC_\gamma
  =
  \left\{
  \bv_\gamma\mathbin{\pm}r_\gamma\bm f_i:
  1\leq i\leq k
  \right\}
  \mathbin{\cup}
  \{\ba_1,\ldots,\ba_d\}.
  \label{eq:orthoplex-range-moving-codebook}
\end{equation}

\begin{proposition}[Multiscale achievability in the nonterminal
orthoplex-bound range]
\label{prop:orthoplex-range-constructive-upper-bound}
For every fixed \(n\ge3\) and \(2\leq k<n\), the vectors in
\eqref{eq:orthoplex-range-auxiliary-gram} and the codebook in
\eqref{eq:orthoplex-range-moving-codebook} exist for all sufficiently
large \(\gamma\).  The codebook is feasible, has \(n+k\) points, and
converges to a deleted cross-polytope having both signs on \(k\) axes
and one sign on each of the remaining \(d=n-k\) axes.  Its maximal
correlation is \(p_\gamma\), and its Gaussian soft-packing energy is
\begin{equation}
\begin{aligned}
  \cG_\gamma(\cC_\gamma)
  &=4k(k-1)\exp\{n\gamma p_\gamma/4\}\\
  &{}+2k\exp\{n\gamma(2p_\gamma-1)/4\}\\
  &{}+d(d+4k-1)\exp\{-n\gamma q_\gamma/4\}.
\end{aligned}
  \label{eq:orthoplex-range-moving-energy}
\end{equation}
Moreover,
\begin{equation}
\begin{aligned}
  n\gamma q_\gamma
  &=\log(n\gamma)\longrightarrow\infty,\\
  n\gamma p_\gamma
  &=O\left(\frac{[\log(n\gamma)]^2}{n\gamma}\right)
  \longrightarrow0.
\end{aligned}
  \label{eq:orthoplex-range-moving-scales}
\end{equation}
Consequently,
\begin{equation}
  \cG_\gamma(\cC_\gamma)\longrightarrow4k(k-1).
  \label{eq:orthoplex-range-moving-energy-limit}
\end{equation}
In particular,
\begin{equation}
  K^{\ast}_{n+k,n}\leq4k(k-1).
  \label{eq:orthoplex-range-optimized-upper-bound}
\end{equation}
The family in \eqref{eq:orthoplex-range-moving-codebook} satisfies
\begin{equation}
  \Pe(\cC_\gamma;\gamma)
  =B_\gamma^{\ast}\bigl(4k(k-1)+o(1)\bigr).
  \label{eq:orthoplex-range-achievable-error}
\end{equation}
\end{proposition}

\begin{IEEEproof}
The proof is given in
Appendix~\ref{app:simplex-orthoplex-application-proofs}.
\end{IEEEproof} 

Proposition~\ref{prop:orthoplex-range-constructive-upper-bound} gives an
achievability family defined for every sufficiently large SNR, rather
than only along a subsequence.  The construction lowers by order
\(q_\gamma\) the correlations of the \(d(d+4k-1)\) pairs involving an auxiliary point that
are closest in the limiting deleted cross-polytope, so their contributions
to the constructed family's leading coefficient vanish because
\(n\gamma q_\gamma\to\infty\).  The correlations of the \(4k(k-1)\) pairs
between shifted points on different axes increase only by
\(p_\gamma=O(q_\gamma^2)\); since \(n\gamma p_\gamma\to0\), these
pairs have limiting weight one.  The exact ML error of the constructed
family therefore has leading coefficient \(4k(k-1)\).  The improvement is obtained only
because the codebook continues to vary with SNR at the specified two
rates.

\subsection{Exact Optimal Leading Coefficients at the Orthoplex-Bound Endpoints}
\label{subsec:orthoplex-exact-endpoints}

For the orthoplex-range endpoints, we first solve the terminal Gaussian
soft-packing energy problem exactly at every positive SNR, which
determines the optimal leading coefficient for \(k=n\), and then provide
a matching lower bound for \(k=2\).

\begin{theorem}[Exact terminal endpoint \(k=n\)]
\label{thm:orthoplex-terminal-endpoint-coefficient}
For every fixed \(n\ge2\) and every \(\gamma>0\), the full
cross-polytope
\[
  \cC_{\mathrm{cross}}
  =\{\be_1,-\be_1,\ldots,\be_n,-\be_n\}
\]
is a global minimizer of the Gaussian soft-packing energy, and
\begin{equation}
  V_\gamma(2n,n)
  =4n(n-1)+2n\exp\{-n\gamma/4\}.
  \label{eq:orthoplex-terminal-endpoint-soft-energy}
\end{equation}
Consequently,
\begin{equation}
  K^{\ast}_{2n,n}
  =K^{\mathrm{fix}}_{2n,n}
  =4n(n-1),
  \label{eq:orthoplex-terminal-endpoint-coefficient}
\end{equation}
and
\begin{equation}
  \Pe^{\ast}(2n,n;\gamma)
  =
  \frac{2(n-1)}{\sqrt{\pi n\gamma}}
  \exp\{-n\gamma/4\}
  \bigl(1+o(1)\bigr).
  \label{eq:orthoplex-terminal-endpoint-awgn-asymptotic}
\end{equation}
The fixed full cross-polytope is therefore asymptotically ML-optimal.
\end{theorem}

\begin{IEEEproof}
The proof is given in
Appendix~\ref{app:simplex-orthoplex-application-proofs}.
\end{IEEEproof} 

The terminal endpoint behaves like the simplex range.  Universal
optimality makes the full cross-polytope a soft-packing energy minimizer at
every positive SNR.  Its \(4n(n-1)\) ordered orthogonal pairs contribute
one unit each to the limiting Gaussian soft-packing energy, whereas the
contributions of the \(2n\) ordered antipodal pairs vanish exponentially.
Hence
\(K^{\mathrm{fix}}_{2n,n}\) and \(K^{\ast}_{2n,n}\) coincide, and SNR-dependent motion cannot improve the
leading coefficient at \(k=n\).  The result also covers the square \(n=k=2\),
where the two orthoplex endpoints coincide.  The exact finite-SNR
assertion concerns the Gaussian soft-packing energy; the exact ML conclusion established
here is asymptotic.

To determine the left endpoint of the orthoplex-bound range, the construction must be paired with a
global converse valid for every codebook, rather than only for a
particular limiting packing.

\begin{proposition}[Global signed-balance lower bound]
\label{prop:orthoplex-left-endpoint-soft-energy-lower-bound}
For every \(\cC\in(\Sph^{n-1})^{n+2}\), put
\(\varrho=\rho_{\max}(\cC)\).  Then, for every \(\gamma>0\),
\begin{equation}
  \cG_\gamma(\cC)
  \ge
  \frac{8}{(1+\varrho)^2}.
  \label{eq:orthoplex-left-endpoint-finite-energy-lower-bound}
\end{equation}
Consequently,
\begin{equation}
  \liminf_{\gamma\to\infty}V_\gamma(n+2,n)\ge8.
  \label{eq:orthoplex-left-endpoint-liminf}
\end{equation}
\end{proposition}
\begin{IEEEproof}
The proof is given in
Appendix~\ref{app:simplex-orthoplex-application-proofs}.
\end{IEEEproof}

Unlike a bound based on a fixed closest-pair count,
Proposition~\ref{prop:orthoplex-left-endpoint-soft-energy-lower-bound}
applies globally to every \((n+2,n)\) codebook, including non-packing-optimal codebooks
and codebooks selected from SNR-dependent families.  Every
packing-tail-competitive family has
\(\rho_{\max}(\cC_\gamma)\to0\), so the finite-SNR inequality forces
the liminf of its Gaussian soft-packing energy to be at least \(8\).  It therefore supplies exactly
the converse needed to certify \(K^{\ast}_{n+2,n}\).

\begin{theorem}[Exact optimal leading coefficient at the left endpoint \(k=2\)]
\label{thm:orthoplex-left-endpoint-coefficient}
For every fixed \(n\ge2\), as \(\gamma\to\infty\),
\begin{equation}
  V_\gamma(n+2,n)=8+o(1),
  \label{eq:orthoplex-left-endpoint-soft-energy-asymptotic}
\end{equation}
and
\begin{equation}
  K^{\ast}_{n+2,n}=8.
  \label{eq:orthoplex-left-endpoint-coefficient}
\end{equation}
The exact SNR-wise optimized ML error satisfies
\begin{equation}
  \Pe^{\ast}(n+2,n;\gamma)
  =
  \frac{\exp\{-n\gamma/4\}}
  {(n+2)\sqrt{n\gamma}\sqrt{\pi}}
  \left(8+o(1)\right).
  \label{eq:orthoplex-left-endpoint-awgn-asymptotic}
\end{equation}
Moreover, a family attaining the optimal leading coefficient is given by the fixed square
when \(n=2\) and by the SNR-dependent codebook in
\eqref{eq:orthoplex-range-moving-codebook}, specialized to \(k=2\),
when \(n\ge3\).  In either case,
\begin{equation}
  \Pe(\cC_\gamma;\gamma)
  =B_\gamma^{\ast}\bigl(8+o(1)\bigr),
  \qquad
  \frac{\Pe(\cC_\gamma;\gamma)}
  {\Pe^{\ast}(n+2,n;\gamma)}
  \longrightarrow1.
  \label{eq:orthoplex-left-endpoint-witness-attainment}
\end{equation}
\end{theorem}

\begin{IEEEproof}
Proposition~\ref{prop:orthoplex-left-endpoint-soft-energy-lower-bound}
gives
\(\liminf_{\gamma\to\infty}V_\gamma(n+2,n)\ge8\).  Let
\(\cC_\gamma\) be the fixed square when \(n=2\) and the family in
\eqref{eq:orthoplex-range-moving-codebook}, specialized to \(k=2\), when
\(n\ge3\).  In the first case,
Theorem~\ref{thm:orthoplex-terminal-endpoint-coefficient} gives
\(\cG_\gamma(\cC_\gamma)=V_\gamma(4,2)=8+4e^{-\gamma/2}\).  In the
second case, \eqref{eq:orthoplex-range-moving-energy-limit} gives
\(\cG_\gamma(\cC_\gamma)\to8\).  Since
\(V_\gamma(n+2,n)\le\cG_\gamma(\cC_\gamma)\), both cases yield the
matching limsup bound.  Hence \(V_\gamma(n+2,n)=8+o(1)\), proving
\eqref{eq:orthoplex-left-endpoint-soft-energy-asymptotic}.
Corollary~\ref{cor:coefficient-extraction}, applied with \(K=8\), then
gives \eqref{eq:orthoplex-left-endpoint-coefficient} and
\eqref{eq:orthoplex-left-endpoint-witness-attainment}.  Substitution of
\eqref{eq:orthoplex-range-packing-scale}, specialized to \(k=2\), gives
\eqref{eq:orthoplex-left-endpoint-awgn-asymptotic}.
\end{IEEEproof}

Theorem~\ref{thm:orthoplex-left-endpoint-coefficient} gives the
strongest proved separation between fixed and SNR-dependent codebook design in
this paper.  The optimal leading coefficient is
\(K^{\ast}_{n+2,n}=8\), independently of \(n\), whereas the best
fixed-codebook leading coefficient is \(K^{\mathrm{fix}}_{n+2,n}=4n\).
For \(n\ge3\), the moving family genuinely attains the optimal leading
coefficient, and its asymptotic ML error is smaller by the factor \(n/2\)
than the best asymptotic ML error of any fixed packing-optimal codebook;
no fixed packing-optimal codebook can attain
\(K^{\ast}_{n+2,n}\).  For \(n=2\), the two endpoints coincide and the fixed
square is already optimal.

\subsection{Optimal Versus Best Fixed-Codebook Leading Coefficients}
\label{subsec:orthoplex-range-fixed-comparison}

Combining Proposition~\ref{prop:orthoplex-range-constructive-upper-bound}
with the fixed benchmark in
\eqref{eq:orthoplex-range-fixed-contact-count} gives, for every fixed
\(n\ge3\) and \(2\leq k<n\),
\begin{equation}
  K^{\ast}_{n+k,n}
  \leq4k(k-1)
  <4n(k-1)
  =K^{\mathrm{fix}}_{n+k,n}.
  \label{eq:orthoplex-range-strict-improvement}
\end{equation}
In particular,
\begin{equation}
  \frac{K^{\ast}_{n+k,n}}
       {K^{\mathrm{fix}}_{n+k,n}}
  \leq\frac{k}{n}<1.
  \label{eq:orthoplex-range-coefficient-ratio-bound}
\end{equation}
Let \(\cC_0^{\mathrm{fix}}\) be any fixed packing-optimal codebook
attaining \(K^{\mathrm{fix}}_{n+k,n}\), and let \(\cC_\gamma\) be the
family in \eqref{eq:orthoplex-range-moving-codebook}.  The fixed-codebook
expansion, \eqref{eq:orthoplex-range-achievable-error}, and
\eqref{eq:optimized-exact-vs-fixed-packing} yield
\begin{equation}
\begin{aligned}
  \frac{\Pe(\cC_\gamma;\gamma)}
       {\Pe(\cC_0^{\mathrm{fix}};\gamma)}
  &\longrightarrow\frac{k}{n},\\
  \frac{\Pe(\cC_0^{\mathrm{fix}};\gamma)}
       {\Pe^{\ast}(n+k,n;\gamma)}
  &\longrightarrow
  \frac{K^{\mathrm{fix}}_{n+k,n}}{K^{\ast}_{n+k,n}}
  \geq\frac{n}{k}.
\end{aligned}
  \label{eq:orthoplex-range-error-comparisons}
\end{equation}
Thus, throughout the nonterminal orthoplex-bound range, the explicit moving
family has a leading coefficient equal to the fraction \(k/n\) of the best
fixed-codebook leading coefficient.  The exact SNR-wise optimum is no larger
than the error of this family.

At \(k=2\),
Theorem~\ref{thm:orthoplex-left-endpoint-coefficient} shows that the
bound is sharp:
\begin{equation}
  \frac{K^{\ast}_{n+2,n}}{K^{\mathrm{fix}}_{n+2,n}}
  =\frac{2}{n},
  \qquad
  \frac{\Pe(\cC_0^{\mathrm{fix}};\gamma)}
       {\Pe^{\ast}(n+2,n;\gamma)}
  \longrightarrow\frac{n}{2}.
  \label{eq:orthoplex-left-endpoint-fixed-moving-comparison}
\end{equation}
For \(3\leq k\leq n-1\), only the one-sided upper bound on
\(K^{\ast}_{n+k,n}\)
has been proved.  At the terminal endpoint,
Theorem~\ref{thm:orthoplex-terminal-endpoint-coefficient} gives
\begin{equation}
  \frac{K^{\ast}_{2n,n}}{K^{\mathrm{fix}}_{2n,n}}=1,
  \label{eq:orthoplex-terminal-endpoint-fixed-comparison}
\end{equation}
so no leading-coefficient gain from an SNR-dependent codebook occurs.
In particular, when \(n=2\), the
left and terminal endpoints coincide and the fixed square is already
leading-order optimal.

For \(3\leq k\leq n-1\), we have not proved a global lower bound
matching \eqref{eq:orthoplex-range-optimized-upper-bound}.  We therefore
cannot exclude a different jointly feasible multiscale deformation with
a still smaller value of the sum in
\eqref{eq:weighted-contact-variational-characterization}.  Nevertheless, the two exact
endpoint results and the structure of the explicit construction suggest
that its achieved value is sharp, motivating the following conjecture.

\begin{conjecture}[Sharpness in the intermediate orthoplex-bound range]
\label{conj:orthoplex-intermediate-sharpness}
For every fixed \(n\ge4\) and \(3\leq k\leq n-1\),
\begin{equation}
  K^{\ast}_{n+k,n}=4k(k-1).
  \label{eq:orthoplex-intermediate-conjecture}
\end{equation}
\end{conjecture}

The conjectured range is empty for \(n\le3\).  If the conjecture holds,
the moving family in \eqref{eq:orthoplex-range-moving-codebook} attains the
optimal leading coefficient throughout the intermediate range.

Table~\ref{tab:coefficient-status} separates the proved exact values,
the proved upper bound, and the conjectured equality.

\begin{table}[!ht]
\caption{Status of the leading coefficients in the simplex and
orthoplex-bound applications.}
\label{tab:coefficient-status}
\centering
\renewcommand{\arraystretch}{1.15}
\small
\begin{tabular}{@{}p{0.22\linewidth}p{0.43\linewidth}p{0.27\linewidth}@{}}
\hline
Range & Optimal leading coefficient & Best fixed-codebook leading coefficient\\
\hline
\(2\le M\le n+1\) & \(M(M-1)\) (exact) & \(M(M-1)\)\\
\(k=2\) & \(8\) (exact) & \(4n\)\\
\(3\le k\le n-1\) &
  \(\le4k(k-1)\); equality conjectured & \(4n(k-1)\)\\
\(k=n\) & \(4n(n-1)\) (exact) & \(4n(n-1)\)\\
\hline
\end{tabular}
\vspace{1mm}

\parbox{\linewidth}{\footnotesize In the intermediate row,
\(4k(k-1)\) is the leading coefficient of the constructed family and
hence an upper bound on \(K^{\ast}_{n+k,n}\); equality is conjectured.}
\end{table}

\section{Conclusion}
\label{sec:conclusion}

Optimal spherical packing fixes the best high-SNR error exponent, but it
need not fix the optimal leading coefficient when the codebook is reoptimized at
each SNR.  The nonterminal orthoplex-bound construction exposes the reason.  It
slightly reduces the separations of selected pairs by an amount negligible
at the leading-coefficient scale, while using the resulting geometric freedom to improve
other pairs that become closest in the limiting packing on a larger,
though still vanishing, scale.  The selected pairs make unit contributions
to the constructed family's leading coefficient, whereas the improved
pairs make vanishing contributions to that coefficient.  The codebooks therefore
approach an optimal packing and retain its exponent, yet outperform every
packing-optimal codebook held fixed.  The improvement results from the two
SNR-dependent convergence rates, not from the limiting packing alone.

For every fixed finite \((M,n)\) with \(M\ge2\) and
\(\rho^{\ast}_{M,n}<1\), we proved
\begin{equation*}
  \Pe^{\ast}(M,n;\gamma)
  =B_\gamma^{\ast}\bigl(K^{\ast}_{M,n}+o(1)\bigr).
\end{equation*}
The best packing-optimal benchmark satisfies
\begin{equation*}
  \Pe^{\mathrm{fix}}(M,n;\gamma)
  =B_\gamma^{\ast}\bigl(K^{\mathrm{fix}}_{M,n}+o(1)\bigr).
\end{equation*}
Moreover, \(K^{\ast}_{M,n}\leq K^{\mathrm{fix}}_{M,n}\).  Thus, whenever
\(K^{\ast}_{M,n}<K^{\mathrm{fix}}_{M,n}\), the exact SNR-wise minimum
error is strictly smaller than the best packing-optimal benchmark for all
sufficiently large \(\gamma\), and their ratio tends to
\(K^{\ast}_{M,n}/K^{\mathrm{fix}}_{M,n}\).  We showed that the minimum of the Gaussian soft-packing
energy converges to \(K^{\ast}_{M,n}\).  We also showed that only pairs that are closest in a limiting packing can
contribute to \(K^{\ast}_{M,n}\) and that their contributions are determined by
the pairs' scaled correlation offsets.

The simplex and orthoplex-bound ranges display both possible outcomes.  For
\(2\leq M\leq n+1\), the regular simplex gives
\(K^{\ast}_{M,n}=K^{\mathrm{fix}}_{M,n}=M(M-1)\), so SNR-dependent motion
provides no leading-order gain.  For \(M=n+k\), \(2\leq k\leq n\), the
best fixed-codebook leading coefficient is \(4n(k-1)\).  When
\(2\leq k<n\), the exact ML error of the explicit moving family has
leading coefficient \(4k(k-1)\), exactly the fraction \(k/n\) of the
best fixed-codebook leading coefficient; hence
\begin{equation*}
  \frac{K^{\ast}_{n+k,n}}{K^{\mathrm{fix}}_{n+k,n}}
  \leq\frac{k}{n}<1.
\end{equation*}
A matching converse proves
\(K^{\ast}_{n+2,n}=8\) at \(k=2\).  At the terminal endpoint \(k=n\), the
fixed cross-polytope attains the optimal leading coefficient and
\(K^{\ast}_{2n,n}=K^{\mathrm{fix}}_{2n,n}=4n(n-1)\).  For
\(3\leq k\leq n-1\), equality in the constructive upper bound remains
conjectural.

\appendices

\section{Uniform Transfer, Optimal Leading-Coefficient Characterization, and
Gaussian Soft-Packing Energy Attainment Proofs}
\label{app:global-optimization-proofs}

This appendix proves the uniform-transfer theorem, the
\(O((n\gamma)^{-1})\) comparison between the normalized optimized exact ML
error and \(V_\gamma(M,n)\), the existence and characterization of the
optimal leading coefficient \(K^{\ast}_{M,n}\), and the Gaussian soft-packing
energy near-minimizer attainment criterion from
Section~\ref{sec:soft-packing-characterization}.
The pointwise fixed-codebook expansion is the standard result cited in
Subsection~\ref{subsec:fixed-codebook-asymptotics} and is not rederived
here.  Instead, the first proof establishes the additional uniformity
needed when the codebook changes with SNR.  It invokes standard
Gaussian-tail and Bonferroni estimates,
while deriving only the codebook-uniform estimates specific to the
present problem.  Set
\begin{align*}
  \delta^{\ast}&=1-\rho^{\ast}_{M,n}>0,\\
  \Delta_e(\cC)&=\rho_e(\cC)-\rho^{\ast}_{M,n},\\
  d_e(\cC)&=1-\rho_e(\cC)=\delta^{\ast}-\Delta_e(\cC).
\end{align*}
For \(e=(m,\ell)\), abbreviate
\(d_{m\ell}:=d_{(m,\ell)}(\cC)\), and write
\begin{equation*}
  w_e(\cC,\gamma)=\exp\{n\gamma\Delta_e(\cC)/4\},
\end{equation*}
so that \(\cG_\gamma(\cC)=\sum_{e\in\cE_M}w_e(\cC,\gamma)\).
Put \(N_M=M(M-1)\).  Hidden constants in \(O_L(\cdot)\) may depend on
the fixed pair \((M,n)\) and on \(L\), but not on \(\cC\) or \(\gamma\).

\subsection{Uniform Transfer on Bounded-Energy Sublevels}

\begin{IEEEproof}[Proof of Theorem~\ref{thm:uniform-soft-packing-transfer}]
Fix \(L\ge1\) and suppose \(\cG_\gamma(\cC)\le L\).  The bounded-energy
packing consequence \eqref{eq:bounded-energy-packing-consequence} gives
\begin{equation*}
  0\le\rho_{\max}(\cC)-\rho^{\ast}_{M,n}
  \le\eta_\gamma,
  \qquad
  \eta_\gamma=\frac{4\log L}{n\gamma}.
\end{equation*}
Hence, for all sufficiently large \(\gamma\),
\begin{equation*}
  \|\bc_m-\bc_\ell\|^2=2d_{m\ell}
  \ge2(\delta^{\ast}-\eta_\gamma)\ge\delta^{\ast},
\end{equation*}
so the codewords are distinct.  If \(M\ge3\), consider the compact set of triples of unit
vectors with all mutual distances at least
\(\sqrt{\delta^{\ast}}\).  On this set, the inner product of the two
unit difference directions from one point is strictly below one:
equality would place the other two points on the same ray, but that ray
meets the unit sphere in only one further point.  Let \(r\) denote the
maximum over this compact set; the preceding argument shows that
\(r<1\).  Put
\(\overline\varrho=\max\{0,r\}\).  Then, uniformly over this bounded-energy
sublevel,
\begin{equation*}
  \bu_{m\ell}^{\T}\bu_{mk}\le\overline\varrho,
  \qquad
  \bu_{m\ell}=\frac{\bc_\ell-\bc_m}
  {\|\bc_\ell-\bc_m\|},
\end{equation*}
for every transmitted index \(m\) and distinct competitors
\(\ell,k\ne m\).  For \(M=2\), this assertion is vacuous.

Define the full ordered-pair tail sum
\begin{equation*}
  \cQ_\gamma(\cC)
  =\frac1M\sum_{e\in\cE_M}
  \Qfun\left(\sqrt{\frac{n\gamma d_e(\cC)}2}\right).
\end{equation*}
The preceding correlation bound gives
\(d_e(\cC)\in[\delta^{\ast}/2,2]\) for all sufficiently large
\(\gamma\).  The classical two-sided Mills inequalities
\cite{Gordon1941} imply, uniformly over all ordered pairs,
\begin{equation*}
\begin{aligned}
  &\frac{M^{-1}\Qfun(\sqrt{n\gamma d_e(\cC)/2})}
  {B_\gamma^{\ast}}\\
  &\quad=w_e(\cC,\gamma)
  \left(\frac{\delta^{\ast}}{d_e(\cC)}\right)^{1/2}\\
  &\qquad\times\left[1+O((n\gamma)^{-1})\right].
\end{aligned}
\end{equation*}
On the stated interval,
\(\lvert(\delta^{\ast}/d_e)^{1/2}-1\rvert
\le C|\Delta_e|\).  For nonnegative offsets,
\begin{equation*}
  \sum_{\Delta_e\ge0}w_e|\Delta_e|
  \le\eta_\gamma\cG_\gamma(\cC)
  =O_L((n\gamma)^{-1}),
\end{equation*}
whereas, for \(\Delta_e=-x<0\),
\(x\exp\{-n\gamma x/4\}\le4/(\mathrm e n\gamma)\).
Together with \(\sum_e w_e\le L\), these estimates give
\begin{equation}
  \frac{\cQ_\gamma(\cC)}{B_\gamma^{\ast}}
  =\cG_\gamma(\cC)+O_L((n\gamma)^{-1}).
  \label{eq:appendix-uniform-pair-sum}
\end{equation}

It remains to replace the sum of pairwise tails by the exact union.  The
union bound and the second Bonferroni inequality
\cite[Ch.~1]{GalambosSimonelli1996} give
\begin{equation*}
\begin{split}
  0&\le \cQ_\gamma(\cC)-\Pe(\cC;\gamma)\\
  &\le\frac1M\sum_{m=1}^M
  \sum_{\substack{\ell<k\\ \ell,k\ne m}}
  \Pp(\cA_{m\ell}\cap \cA_{mk}).
\end{split}
\end{equation*}
Each event has the form
\(\cA_{m\ell}=\{\bZ^{\T}\bu_{m\ell}\ge x_{m\ell}\}\), where
\begin{equation*}
  x_{m\ell}=\sqrt{\frac{n\gamma d_{m\ell}}2}
  \ge x_{\min}
  :=\sqrt{\frac{n\gamma(\delta^{\ast}-\eta_\gamma)}2}.
\end{equation*}
If the two normals are antipodal, the corresponding events are
disjoint.  Otherwise their intersection implies
\(\bZ^{\T}(\bu_{m\ell}+\bu_{mk})\ge2x_{\min}\), and the standard scalar
Gaussian tail bound \cite{Proakis2001} gives
\begin{equation*}
  \Pp(\cA_{m\ell}\cap \cA_{mk})
  \le\exp\left\{-\frac{n\gamma(\delta^{\ast}-\eta_\gamma)}
  {2(1+\overline\varrho)}\right\}.
\end{equation*}
Because \(\overline\varrho<1\), this exponent exceeds the exponent
\(n\gamma\delta^{\ast}/4\) in \(B_\gamma^{\ast}\) by a fixed positive
multiple of \(n\gamma\) for all sufficiently large \(\gamma\).  The
number of intersections is fixed, so, for some \(\chi_L>0\),
\begin{equation*}
  \frac{\cQ_\gamma(\cC)-\Pe(\cC;\gamma)}{B_\gamma^{\ast}}
  =O_L(\exp\{-\chi_L n\gamma\}).
\end{equation*}
Combining this estimate with
\eqref{eq:appendix-uniform-pair-sum} proves the asymptotic form of
\eqref{eq:uniform-soft-packing-transfer}.  Enlarging its constant gives
the asserted bound for every \(\gamma\ge\gamma_L\).
\end{IEEEproof}

\subsection{Packing-Scale Error Implies Bounded Gaussian Soft-Packing Energy}

\begin{lemma}[Error-to-energy localization]
\label{lem:error-to-energy-localization}
If an SNR-indexed family \(\{\cC_\gamma\}\subset\cM\) satisfies
\begin{equation*}
  \frac{\Pe(\cC_\gamma;\gamma)}{B_\gamma^{\ast}}=O(1),
\end{equation*}
then its codewords are pairwise distinct for all sufficiently large
\(\gamma\), and
\begin{equation}
  \rho_{\max}(\cC_\gamma)-\rho^{\ast}_{M,n}
  =O((n\gamma)^{-1}),
  \qquad
  \cG_\gamma(\cC_\gamma)=O(1).
  \label{eq:error-to-energy-localization}
\end{equation}
\end{lemma}

\begin{IEEEproof}
If two codewords coincide, then for every received vector a decoder can
correctly distinguish at most one of their two equiprobable message
labels; equivalently, the two corresponding conditional probabilities
of correct decoding sum to at most one, independently of the tie-breaking
rule.  Those two messages therefore contribute at least \(1/M\) to the
average error.  Since \(B_\gamma^{\ast}\to0\), the assumed packing-scale
error makes coincident codewords impossible for all sufficiently large
\(\gamma\).  Put
\(\xi_\gamma=\rho_{\max}(\cC_\gamma)-\rho^{\ast}_{M,n}\ge0\) and
\(d_\gamma=\delta^{\ast}-\xi_\gamma\).  A pair attaining the maximal
correlation gives
\begin{equation*}
  \Pe(\cC_\gamma;\gamma)
  \ge\frac1M\Qfun\left(\sqrt{\frac{n\gamma d_\gamma}{2}}\right).
\end{equation*}
The global lower Gaussian-tail estimate
\begin{equation*}
  \Qfun(x)\ge c\frac{\exp\{-x^2/2\}}{1+x},
  \qquad x\ge0,
\end{equation*}
follows, for example, by combining monotonicity on \(0\le x\le1\)
with the classical Mills inequalities on \(x\ge1\)
\cite{Gordon1941}.  Since \(0\le d_\gamma\le\delta^{\ast}\), it gives
\begin{equation*}
\begin{aligned}
  \frac{\Pe(\cC_\gamma;\gamma)}{B_\gamma^{\ast}}
  &\ge c'\frac{\sqrt{n\gamma}}
  {1+\sqrt{n\gamma d_\gamma/2}}\\
  &\qquad\times\exp\{n\gamma\xi_\gamma/4\}\\
  &\ge c''\exp\{n\gamma\xi_\gamma/4\}.
\end{aligned}
\end{equation*}
Bounded normalized error therefore implies
\(n\gamma\xi_\gamma=O(1)\).  Since every offset
\(\Delta_e(\cC_\gamma)\le\xi_\gamma\),
\begin{equation*}
  \cG_\gamma(\cC_\gamma)
  \le N_M\exp\{n\gamma\xi_\gamma/4\}=O(1),
\end{equation*}
which proves \eqref{eq:error-to-energy-localization}.
\end{IEEEproof}

\subsection{Equivalence of the Optimized Values}

\begin{IEEEproof}[Proof of Corollary~\ref{cor:optimized-soft-packing-transfer}]
Let \(\cC_\gamma^{\mathsf G}\) minimize \(\cG_\gamma\).  Evaluation at
any fixed packing-optimal codebook gives \(V_\gamma(M,n)\le N_M\), so
Theorem~\ref{thm:uniform-soft-packing-transfer} yields
\begin{equation*}
  \frac{\Pe^{\ast}(M,n;\gamma)}{B_\gamma^{\ast}}
  \le\frac{\Pe(\cC_\gamma^{\mathsf G};\gamma)}{B_\gamma^{\ast}}
  =V_\gamma(M,n)+O((n\gamma)^{-1}).
\end{equation*}

For the reverse inequality, choose an exact ML minimizer
\(\cC_\gamma^{\ast}\) as in
Subsection~\ref{subsec:exact-design-problem}.  Comparison with any fixed
\(\cC_0\in\cP^{\ast}\), followed by
\eqref{eq:fixed-packing-normalized-limit}, gives
\begin{equation*}
  \frac{\Pe(\cC_\gamma^{\ast};\gamma)}{B_\gamma^{\ast}}
  \le\frac{\Pe(\cC_0;\gamma)}{B_\gamma^{\ast}}=O(1).
\end{equation*}
Lemma~\ref{lem:error-to-energy-localization} shows that
\(\cG_\gamma(\cC_\gamma^{\ast})\) is bounded.  Uniform transfer can
therefore be applied with one fixed sublevel, giving
\begin{equation*}
\begin{aligned}
  \frac{\Pe^{\ast}(M,n;\gamma)}{B_\gamma^{\ast}}
  &=\cG_\gamma(\cC_\gamma^{\ast})+O((n\gamma)^{-1})\\
  &\ge V_\gamma(M,n)-O((n\gamma)^{-1}).
\end{aligned}
\end{equation*}
The two bounds prove \eqref{eq:optimized-soft-packing-transfer}.
\end{IEEEproof}

\subsection{Existence and Characterization of the Optimal Leading Coefficient}

\begin{IEEEproof}[Proof of Theorem~\ref{thm:global-coefficient}]
The compact codebook space \(\cM\) is semialgebraic, and
\((\gamma,\cC)\mapsto\cG_\gamma(\cC)\) is definable, with real
parameters, in the real exponential field \(\R_{\exp}\).  Closure under
quantification over \(\cM\) therefore makes the value function
\(V_\gamma(M,n)\) definable.  O-minimality of \(\R_{\exp}\)
\cite{Wilkie1996} and the one-variable monotonicity theorem
\cite[Ch.~3, Sec.~1]{VanDenDries1998} imply that this function is monotone
or constant on some interval \((\gamma_0,\infty)\).

For every codebook, \(\rho_{\max}(\cC)\ge\rho^{\ast}_{M,n}\).  If the
unordered pair \(\{m,\ell\}\) attains \(\rho_{\max}(\cC)\), then both
ordered terms \((m,\ell)\) and \((\ell,m)\) in \(\cG_\gamma(\cC)\) are
at least one.  Thus \(V_\gamma(M,n)\ge2\).  Conversely, evaluating the
energy at a packing-optimal codebook gives
\(V_\gamma(M,n)\le N_M\), because every packing gap is nonpositive.
Eventual monotonicity and these uniform bounds prove existence and
finiteness of
\(\lim_{\gamma\to\infty}V_\gamma(M,n)\), together with
\eqref{eq:Kstar-elementary-bounds}.

For a fixed \(\cC_0\in\cP^{\ast}\),
equation~\eqref{eq:soft-packing-fixed-limit} and
\(V_\gamma\le\cG_\gamma(\cC_0)\) give
\eqref{eq:Kstar-fixed-packing-upper}.  Finally,
Corollary~\ref{cor:optimized-soft-packing-transfer} gives
\begin{equation*}
  \frac{\Pe^{\ast}(M,n;\gamma)}{B_\gamma^{\ast}}
  =V_\gamma(M,n)+O((n\gamma)^{-1})
  =K^{\ast}_{M,n}+o(1),
\end{equation*}
which identifies the second limit in
\eqref{eq:Kstar-soft-packing-limit} and proves \eqref{eq:global-main}.
\end{IEEEproof}

\subsection{Attainment by Gaussian Soft-Packing Energy Near-Minimization}

\begin{IEEEproof}[Proof of Corollary~\ref{cor:soft-energy-minimizers-attain}]
Equations~\eqref{eq:soft-energy-near-minimizer-gap} and
\eqref{eq:Kstar-soft-packing-limit} give
\(\cG_\gamma(\widehat{\cC}_\gamma)\to K^{\ast}_{M,n}\); the bound
\(V_\gamma(M,n)\le N_M\) shows that this family has uniformly bounded
Gaussian soft-packing energy.  Uniform
transfer then gives
\(\Pe(\widehat{\cC}_\gamma;\gamma)/B_\gamma^{\ast}\to
K^{\ast}_{M,n}\).  Dividing by the optimized normalized-error limit in
Theorem~\ref{thm:global-coefficient} is legitimate because
\(K^{\ast}_{M,n}\ge2\), and proves
\eqref{eq:soft-energy-near-minimizer-relative-attainment}.
\end{IEEEproof}

\section{Packing-Tail Localization and Closest-Pair Contribution
Characterization}
\label{app:local-reduction-proofs}

This appendix proves
Lemma~\ref{lem:coefficient-attaining-localization},
Theorem~\ref{thm:packing-contact-reduction}, and
Proposition~\ref{prop:weighted-contact-variational-characterization}
from Section~\ref{sec:geometric-reduction}.
Appendix~\ref{app:global-optimization-proofs}
already supplies error-to-energy localization, uniform transfer, and
convergence of the globally minimized Gaussian soft-packing energy.  The first two proofs
use compactness and separation of finitely many pair correlations; the
final proof combines the limiting closest-pair reduction with the global
soft-packing energy limit.

At a coarser level, the known Riesz-energy fact that cluster points of
fixed-cardinality minimizers are best packings is analogous to the first
localization step \cite[Prop.~2]{BondarenkoHardinSaff2014}.  It does not
provide the \(1/(n\gamma)\) localization rate, the scaled closest-pair
contribution formula, or the exact ML conclusions proved here.

\begin{IEEEproof}[Proof of Lemma~\ref{lem:coefficient-attaining-localization}]
Condition~\eqref{eq:packing-tail-competitive-family} and
Lemma~\ref{lem:error-to-energy-localization} give
\eqref{eq:attaining-family-packing-gap} and
\(\cG_\gamma(\cC_\gamma)=O(1)\), and hence
\(\rho_{\max}(\cC_\gamma)\to\rho^{\ast}_{M,n}\).  If the distance to
\(\cP^{\ast}\) did not tend to zero, compactness would give a
convergent subsequence with fixed positive separation from
\(\cP^{\ast}\), whereas continuity of \(\rho_{\max}\) would place its
limit in \(\cP^{\ast}\).  This proves
\eqref{eq:attaining-family-Pstar-localization}.

Under the additional hypotheses, membership
\(\cC_{\gamma_j}\in\cP^{\ast}\) along an unbounded subsequence would
give
\begin{equation*}
  \cG_{\gamma_j}(\cC_{\gamma_j})
  \ge |\cI_{\max}(\cC_{\gamma_j})|
  \ge K^{\mathrm{fix}}_{M,n}.
\end{equation*}
The first inequality follows directly from
\eqref{eq:global-soft-packing-energy}, because every ordered closest-pair
index contributes one to \(\cG_{\gamma_j}(\cC_{\gamma_j})\); the second is the fixed closest-pair-count characterization
\eqref{eq:fixed-packing-coefficient}.
Uniform transfer would then give
\begin{equation*}
  \frac{\Pe(\cC_{\gamma_j};\gamma_j)}{B_{\gamma_j}^{\ast}}
  \ge K^{\mathrm{fix}}_{M,n}-O((n\gamma_j)^{-1}),
\end{equation*}
while attainment of the optimal leading coefficient makes this ratio converge to
\(K^{\ast}_{M,n}<K^{\mathrm{fix}}_{M,n}\), a contradiction.
This proves \eqref{eq:attaining-family-eventually-outside-Pstar}.
\end{IEEEproof}

\begin{IEEEproof}[Proof of Theorem~\ref{thm:packing-contact-reduction}]
Finiteness and \(\cC_\gamma\to\cC_\infty\) give \(a>0\) such that,
for all sufficiently large \(\gamma\), every
\(e\notin\cI_{\max}(\cC_\infty)\) satisfies
\(\rho_e(\cC_\gamma)-\rho^{\ast}_{M,n}\le-a\) (vacuously if there is
no such index).  Since \(\rho_{\max}(\cC_\gamma)\ge\rho^{\ast}_{M,n}\),
the largest correlation must therefore be attained within
\(\cI_{\max}(\cC_\infty)\), and
\begin{equation*}
  \max_{e\in\cI_{\max}(\cC_\infty)}\Delta_{e,\gamma}
  =\rho_{\max}(\cC_\gamma)-\rho^{\ast}_{M,n}\ge0.
\end{equation*}
Equation~\eqref{eq:bounded-energy-packing-consequence} gives the stated
upper bound.  Separating the indices in
\(\cI_{\max}(\cC_\infty)\) from the remaining pairs gives
\begin{equation*}
\begin{split}
  \cG_\gamma(\cC_\gamma)
  &={}
  \sum_{e\in\cI_{\max}(\cC_\infty)}
  \exp\{n\gamma\Delta_{e,\gamma}/4\}\\
  &\quad+O(\exp\{-a n\gamma/4\}).
\end{split}
\end{equation*}
The remainder is identically zero if every ordered pair belongs to
\(\cI_{\max}(\cC_\infty)\).  Thus
\eqref{eq:packing-contact-energy-reduction} holds for some \(c>0\).
Uniform transfer and absorption of the exponential remainder into
\(O((n\gamma)^{-1})\) then prove
\eqref{eq:packing-contact-error-reduction}.
\end{IEEEproof}

\subsection{Global Variational Characterization by Scaled Closest-Pair Offsets}

\begin{IEEEproof}[Proof of
Proposition~\ref{prop:weighted-contact-variational-characterization}]
Fix \(\cC_0\in\cP^{\ast}\),
\(\bm\lambda\in\mathfrak L(\cC_0)\), and a witnessing sequence from
\eqref{eq:admissible-contact-limit-set}.  Theorem~\ref{thm:packing-contact-reduction},
specifically \eqref{eq:packing-contact-energy-reduction}, gives
\begin{equation*}
  \cG_{\gamma_j}(\cC_j)
  \longrightarrow
  \sum_{e\in\cI_{\max}(\cC_0)}\exp\{\lambda_e/4\}.
\end{equation*}
Since \(V_{\gamma_j}(M,n)\le\cG_{\gamma_j}(\cC_j)\) and
Theorem~\ref{thm:global-coefficient} gives
\(V_{\gamma_j}(M,n)\to K^{\ast}_{M,n}\), the sum in
\eqref{eq:weighted-contact-variational-characterization} is at least
\(K^{\ast}_{M,n}\) for every admissible pair.

For the reverse inequality, start with \(\widetilde\gamma_j=j\) and
select an exact soft-packing energy minimizer
\(\widetilde\cC_j^{\mathsf G}\) at each \(\widetilde\gamma_j\).  Its
energy is
\begin{equation*}
  \cG_{\widetilde\gamma_j}(\widetilde\cC_j^{\mathsf G})
  =V_{\widetilde\gamma_j}(M,n)\le M(M-1),
\end{equation*}
so compactness of \(\cM\) gives indices \(j_r\to\infty\) for which
\(\widetilde\cC_{j_r}^{\mathsf G}\to\cC_0\).  Relabel this subsequence by
\(\gamma_r=\widetilde\gamma_{j_r}=j_r\) and
\(\cC_r^{\mathsf G}=\widetilde\cC_{j_r}^{\mathsf G}\).
Equation~\eqref{eq:bounded-energy-packing-consequence}
and continuity of \(\rho_{\max}\) imply
\(\cC_0\in\cP^{\ast}\).  Corollary~\ref{cor:subsequential-weighted-contact}
then gives a further subsequence and
\(\bm\lambda\in\mathfrak L(\cC_0)\) such that
\begin{equation*}
  \sum_{e\in\cI_{\max}(\cC_0)}\exp\{\lambda_e/4\}
  =\lim_{r\to\infty}V_{\gamma_r}(M,n)
  =K^{\ast}_{M,n}.
\end{equation*}
This proves both equality and attainment in
\eqref{eq:weighted-contact-variational-characterization}.  Finally, for
any sequence witnessing a minimizing pair,
\(\cG_{\gamma_j}(\cC_j)\to K^{\ast}_{M,n}\), and the Gaussian
soft-packing energy is uniformly bounded along that sequence.
Theorem~\ref{thm:uniform-soft-packing-transfer}
then gives the limit of the normalized ML error.  Dividing it by the
corresponding limit of
\(\Pe^{\ast}(M,n;\gamma)/B_\gamma^{\ast}\) from
Theorem~\ref{thm:global-coefficient}, which is positive because
\(K^{\ast}_{M,n}\ge2\), gives the relative-error limit.
\end{IEEEproof}

\section{Proofs for the Simplex and Orthoplex-Bound Applications}
\label{app:simplex-orthoplex-application-proofs}

This appendix supplies the derivations specific to the simplex and
orthoplex-bound applications.  The simplex and terminal-endpoint arguments
use the universal-optimality result in
\cite[Thm.~1.2 and Table~1]{CohnKumar2007}, while the orthoplex-bound
fixed-codebook argument uses the equality-case classification in
\cite[Thm.~3]{Kuperberg2007}; an equivalent characterization is given in
\cite[Thm.~2.2]{KingMixonParshallWells2026}.
Only the normalization and closest-pair-count consequences needed here are
derived.  The multiscale moving-family and signed-balance proofs are
paper-specific and are retained in full.

\subsection{Simplex Gaussian Soft-Packing Energy Minimum}

\begin{IEEEproof}[Proof of Proposition~\ref{prop:simplex-energy-minimum}]
Put \(t=n\gamma/4\) and \(f_t(s)=\exp\{-ts/2\}\).  Since
\begin{equation*}
  \cG_\gamma(\cC)
  =\exp\left\{\frac{tM}{M-1}\right\}
  \sum_{m\ne\ell}f_t(\|\bc_m-\bc_\ell\|^2),
\end{equation*}
and \(f_t\) is strictly completely monotone, the universal-optimality
result in
\cite[Thm.~1.2, Table~1, and the simplex special case]{CohnKumar2007}
makes the embedded regular simplex the unique minimizer among distinct
\(M\)-point configurations.
For \(n\ge2\), continuity extends the same lower bound to the full
labeled tuple space, including coincident codewords.

For completeness, the equality case on this enlarged space follows
from the centroid identity and strict Jensen convexity:
\begin{equation*}
  \sum_{m\ne\ell}\bc_m^{\mathsf T}\bc_\ell
  =\left\|\sum_{m=1}^M\bc_m\right\|^2-M\ge-M.
\end{equation*}
Thus, with \(N_M=M(M-1)\),
\begin{equation*}
  \cG_\gamma(\cC)
  \ge
  N_M\exp\left\{t\left(
  \frac{1}{N_M}\sum_{m\ne\ell}\bc_m^{\mathsf T}\bc_\ell
  +\frac{1}{M-1}\right)\right\}
  \ge N_M.
\end{equation*}
Attaining the value \(M(M-1)\) forces zero centroid and equal
off-diagonal correlations, hence correlation \(-1/(M-1)\) for every
ordered pair.  Thus the minimizers are exactly the embedded regular
simplices, up to orthogonal transformation and relabeling.  The case
\(n=1,M=2\) follows directly by comparing the antipodal and coincident
labeled pairs.
\end{IEEEproof}

\subsection{Best Fixed-Codebook Leading Coefficient in the Orthoplex-Bound Range}

\begin{IEEEproof}[Proof of Proposition~\ref{prop:orthoplex-range-fixed-contact}]
Fix \(2\leq k\leq n\), put \(M=n+k\), and let
\(\cC_0=\{\bc_1,\ldots,\bc_M\}\in\cP^{\ast}\).  Since
\(\rho^{\ast}_{M,n}=0\), all off-diagonal correlations are
nonpositive.  Form the graph on \([M]\) in which \(m\) and \(\ell\)
are adjacent when \(\bc_m^{\mathsf T}\bc_\ell<0\).  Vectors belonging
to different connected components are orthogonal.

For a connected component of size \(s\), let \(r\) be the rank of its
Gram matrix \(\bG\).  The matrix \(\bA=\bI_s-\bG\) is nonnegative and
irreducible.  Since \(\bG\succeq0\), the largest eigenvalue of \(\bA\) is
at most one.  If the component is linearly dependent, then one is an
eigenvalue of \(\bA\), and the Perron--Frobenius theorem shows that it is
simple.  Thus every component has nullity either zero or one.  This is
the connected-component form of the equality-case classification in
\cite[Thm.~3]{Kuperberg2007}; see also
\cite[Thm.~2.2]{KingMixonParshallWells2026}.

Let \(\ell\) be the number of dependent components.  Since the full
Gram matrix has rank at most \(n\),
\[
  \ell\geq M-n=k.
\]
Write the ranks of the dependent components as
\(r_1,\ldots,r_\ell\geq1\), so that their sizes are
\(r_1+1,\ldots,r_\ell+1\).  Let \(r_0\geq0\) be the total size of the
remaining, full-rank components.  Then
\begin{equation*}
  r_0+\sum_{i=1}^{\ell}(r_i+1)=M.
\end{equation*}
If \(N_-\) denotes the number of ordered pairs having strictly
negative correlation, then
\begin{equation*}
  N_-
  \leq r_0(r_0-1)+\sum_{i=1}^{\ell}r_i(r_i+1).
\end{equation*}
Indeed, negative-correlation pairs can occur only within connected
components, and the right-hand side bounds all ordered pairs within
those components.

Merging \(r_0\) into any one of the \(r_i\)'s can only increase the
right-hand side, because
\begin{equation*}
  (r_0+r_i)(r_0+r_i+1)
  -r_0(r_0-1)-r_i(r_i+1)
  =2r_0(r_i+1)\geq0.
\end{equation*}
For \(\ell\) positive integers with sum \(n+k-\ell\), convexity then
shows that the sum of \(r(r+1)\) is maximized when one of them equals
\(n+k-2\ell+1\) and the remaining \(\ell-1\) equal one.  The resulting
upper bound
\begin{equation*}
  (n+k-2\ell+1)(n+k-2\ell+2)+2(\ell-1)
\end{equation*}
is decreasing over the feasible integers \(\ell\geq k\).  Therefore,
\begin{equation*}
  N_-
  \leq(n-k+1)(n-k+2)+2(k-1).
\end{equation*}
Every off-diagonal correlation is either negative or zero, and the
ordered zero-correlation pairs are precisely the ordered closest-pair
indices.  Hence
\begin{align*}
  |\cI_{\max}(\cC_0)|
  &=M(M-1)-N_-\\
  &\geq(n+k)(n+k-1)
  -(n-k+1)(n-k+2)-2(k-1)\\
  &=4n(k-1).
\end{align*}

Conversely, take a regular \((n-k+1)\)-simplex and \(k-1\) antipodal
pairs in mutually orthogonal subspaces.  All correlations within each
block are strictly negative, all cross-block correlations are zero,
and the ordered zero-correlation count is exactly \(4n(k-1)\).  The
codebook is packing-optimal, so minimizing as in
\eqref{eq:fixed-packing-coefficient} proves
\eqref{eq:orthoplex-range-fixed-contact-count}.
\end{IEEEproof}

\subsection{Multiscale Construction in the Nonterminal Orthoplex-Bound Range}

\begin{IEEEproof}[Proof of
Proposition~\ref{prop:orthoplex-range-constructive-upper-bound}]
Fix \(n\ge3\) and \(2\le k<n\), and put \(d=n-k\).  The Gram matrix
specified by \eqref{eq:orthoplex-range-auxiliary-gram} is
\[
  (1+q_\gamma)\bI-q_\gamma\bm 1\bm 1^{\mathsf T}.
\]
Its eigenvalues are \(1+q_\gamma\), with multiplicity \(d-1\), and
\(1-(d-1)q_\gamma\), with multiplicity one.  It is therefore positive
definite for all sufficiently large \(\gamma\).  Taking its
positive-definite square root in a fixed basis of \(\mathcal U\)
produces the continuous realization specified in the body, and
\(\ba_1,\ldots,\ba_d\) converge to that basis.

For each \(j\),
\[
  \ba_j^{\mathsf T}\sum_{i=1}^d\ba_i
  =1-(d-1)q_\gamma.
\]
The definition of \(\bv_\gamma\) therefore gives
\[
  \ba_j^{\mathsf T}\bv_\gamma=-q_\gamma,
  \qquad
  \|\bv_\gamma\|^2
  =\frac{dq_\gamma^2}{1-(d-1)q_\gamma}
  =p_\gamma.
\]
Because \(\bv_\gamma\in\mathcal U\) and
\(\bm f_i\in\mathcal U^\perp\),
\[
  \|\bv_\gamma\pm r_\gamma\bm f_i\|^2
  =p_\gamma+r_\gamma^2=1.
\]
For sufficiently large \(\gamma\), \(p_\gamma<1/2\), so the points in
\eqref{eq:orthoplex-range-moving-codebook} are distinct and form a
spherical codebook of size \(2k+d=n+k\).  Since
\(\bv_\gamma\to\bzero\), \(r_\gamma\to1\), and the auxiliary vectors
converge to an orthonormal basis of \(\mathcal U\), this codebook
converges to the stated deleted cross-polytope.

Write a shifted point as
\(\bx_{i,s}=\bv_\gamma+s r_\gamma\bm f_i\), where
\(s\in\{-1,1\}\).  Its pair correlations are
\[
\begin{array}{@{}lcc@{}}
\text{pair class} & \text{correlation} & \text{ordered count}\\
\hline
\text{shifted points on different axes}
  & p_\gamma & 4k(k-1)\\
\text{opposite points on one shifted axis}
  & 2p_\gamma-1 & 2k\\
\text{pairs involving auxiliary points}
  & -q_\gamma & d(d+4k-1).
\end{array}
\]
Indeed,
\(\bx_{i,s}^{\mathsf T}\bx_{j,t}=p_\gamma\) for \(i\ne j\),
\(\bx_{i,1}^{\mathsf T}\bx_{i,-1}=2p_\gamma-1\), and every
off-diagonal pair containing an auxiliary point has correlation
\(-q_\gamma\).  The last count is the sum of \(d(d-1)\) ordered
auxiliary--auxiliary pairs and \(4kd\) ordered
auxiliary--shifted pairs.  The counts satisfy
\begin{equation}
  4k(k-1)+2k+d(d+4k-1)
  =(n+k)(n+k-1).
  \label{eq:orthoplex-range-pair-count-check}
\end{equation}

For every sufficiently large finite \(\gamma\),
\(p_\gamma>0>-q_\gamma\) and \(2p_\gamma-1<0\).  Hence
\(\rho_{\max}(\cC_\gamma)=p_\gamma\), attained by the
\(4k(k-1)\) ordered cross-axis pairs.  Since
\eqref{eq:orthoplex-range-packing-correlation} gives
\(\rho^{\ast}_{n+k,n}=0\), summing the three pair classes proves
\eqref{eq:orthoplex-range-moving-energy}.

The definition of \(q_\gamma\) gives
\(n\gamma q_\gamma=\log(n\gamma)\), while
\[
  n\gamma p_\gamma
  =
  \frac{d[\log(n\gamma)]^2}
  {n\gamma[1-(d-1)q_\gamma]}
  =
  O\left(\frac{[\log(n\gamma)]^2}{n\gamma}\right).
\]
This proves \eqref{eq:orthoplex-range-moving-scales}.  In
\eqref{eq:orthoplex-range-moving-energy}, the first term therefore
tends to \(4k(k-1)\), while the second and third terms vanish.  Thus
\eqref{eq:orthoplex-range-moving-energy-limit} holds.

Finally,
\(V_\gamma(n+k,n)\le\cG_\gamma(\cC_\gamma)\).  Taking limits and using
Theorem~\ref{thm:global-coefficient} proves
\eqref{eq:orthoplex-range-optimized-upper-bound}.  The constructed
family has uniformly bounded Gaussian soft-packing energy, so
Theorem~\ref{thm:uniform-soft-packing-transfer} and
\eqref{eq:orthoplex-range-moving-energy-limit} give
\eqref{eq:orthoplex-range-achievable-error}.
\end{IEEEproof}

\subsection{Terminal Orthoplex Endpoint}

\begin{IEEEproof}[Proof of
Theorem~\ref{thm:orthoplex-terminal-endpoint-coefficient}]
Put \(t=n\gamma/4\).  Since \(\rho^{\ast}_{2n,n}=0\),
\begin{equation*}
  \cG_\gamma(\cC)
  =\sum_{m\ne\ell}\exp\{t\bc_m^{\mathsf T}\bc_\ell\}.
\end{equation*}
With \(f_t(s)=\exp\{-ts/2\}\), one has
\(\exp\{t\bx^{\mathsf T}\by\}=e^t f_t(\|\bx-\by\|^2)\), and
\(f_t\) is strictly completely monotone.  Universal optimality of the
full cross-polytope therefore makes it a global soft-packing energy minimizer
\cite[Thm.~1.2 and Table~1]{CohnKumar2007}; continuity extends the
lower bound to the labeled tuple space allowing coincidences.

The cross-polytope has \(4n(n-1)\) ordered orthogonal pairs and \(2n\)
ordered antipodal pairs.  Their respective contributions to
\(\cG_\gamma(\cC_{\mathrm{cross}})\) are \(1\) and \(e^{-t}\),
respectively, proving
\eqref{eq:orthoplex-terminal-endpoint-soft-energy}.
Theorem~\ref{thm:global-coefficient} and
Proposition~\ref{prop:orthoplex-range-fixed-contact}, specialized to
\(k=n\), give \eqref{eq:orthoplex-terminal-endpoint-coefficient}.
Corollary~\ref{cor:optimized-soft-packing-transfer} together with
\eqref{eq:orthoplex-range-packing-scale}, again specialized to \(k=n\),
gives \eqref{eq:orthoplex-terminal-endpoint-awgn-asymptotic}.

Finally, the fixed cross-polytope is an exact soft-packing energy minimizer at
every SNR.  Corollary~\ref{cor:soft-energy-minimizers-attain}, applied
with zero objective gap, gives
\begin{equation*}
  \frac{\Pe(\cC_{\mathrm{cross}};\gamma)}
  {\Pe^{\ast}(2n,n;\gamma)}\longrightarrow1,
\end{equation*}
which proves relative-error optimality of this particular fixed
codebook.
\end{IEEEproof}

\subsection{Signed-Balance Converse at the Left Endpoint}

\begin{IEEEproof}[Proof of
Proposition~\ref{prop:orthoplex-left-endpoint-soft-energy-lower-bound}]
By \eqref{eq:orthoplex-range-packing-correlation}, specialized to
\(k=2\), the normalization in
\(\cG_\gamma\) has \(\rho^{\ast}_{n+2,n}=0\).  Put
\(t=n\gamma/4\) and \(\varrho=\rho_{\max}(\cC)\).  If two codewords
coincide, then \(\varrho=1\), and their two ordered terms give
\[
  \cG_\gamma(\cC)\ge2\exp\{t\}
  \ge2=\frac{8}{(1+\varrho)^2}.
\]
We may therefore assume that the codewords are distinct.

By Radon's theorem \cite{Radon1921}, after discarding zero weights there
are disjoint nonempty index sets \(\mathcal J_+\) and
\(\mathcal J_-\), positive probability weights
\((\alpha_i)_{i\in\mathcal J_+}\) and
\((\beta_j)_{j\in\mathcal J_-}\), and \(\bm\mu\in\R^n\) such that
\[
  \sum_{i\in\mathcal J_+}\alpha_i\bc_i
  =\sum_{j\in\mathcal J_-}\beta_j\bc_j
  =\bm\mu.
\]
For every \(i\in\mathcal J_+\), the \(\mathcal J_+\) representation and
\(\bc_i^{\mathsf T}\bc_h\ge-1\) give
\begin{equation*}
  \bc_i^{\mathsf T}\bm\mu
  =\alpha_i+
  \sum_{\substack{h\in\mathcal J_+\\h\ne i}}
  \alpha_h\bc_i^{\mathsf T}\bc_h
  \ge\alpha_i-(1-\alpha_i)
  =2\alpha_i-1.
\end{equation*}
The \(\mathcal J_-\) representation instead gives
\begin{equation*}
  \bc_i^{\mathsf T}\bm\mu
  =\sum_{j\in\mathcal J_-}\beta_j
  \bc_i^{\mathsf T}\bc_j
  \le\varrho.
\end{equation*}
Interchanging the two sets gives the analogous bound for each
\(\beta_j\).  Hence
\[
  \|\bm\alpha\|_\infty,\ \|\bm\beta\|_\infty
  \le\frac{1+\varrho}{2}.
\]

Since \(\sum_{i,j}\alpha_i\beta_j=1\), weighted AM--GM applied to
\begin{equation*}
  \frac{\exp\{t\bc_i^{\mathsf T}\bc_j\}}
  {\alpha_i\beta_j}
  \quad\text{with weight }\alpha_i\beta_j
\end{equation*}
yields
\begin{align*}
  S
  &:={}
  \sum_{i\in\mathcal J_+}\sum_{j\in\mathcal J_-}
  \exp\{t\bc_i^\T\bc_j\}\\
  &\ge
  \frac{\exp\{t\|\bm\mu\|^2\}}
  {\displaystyle
   \prod_{i\in\mathcal J_+}\alpha_i^{\alpha_i}
   \prod_{j\in\mathcal J_-}\beta_j^{\beta_j}}\\
  &\ge
  \frac{1}{\|\bm\alpha\|_\infty\|\bm\beta\|_\infty}
  \ge\frac{4}{(1+\varrho)^2}.
\end{align*}
Here \(\exp\{t\|\bm\mu\|^2\}\ge1\).  Moreover,
\(\alpha_i\le\|\bm\alpha\|_\infty\) and
\(\sum_i\alpha_i=1\) give
\begin{equation*}
  \prod_i\alpha_i^{\alpha_i}
  \le\|\bm\alpha\|_\infty,
\end{equation*}
and the analogous inequality holds for \(\bm\beta\).
The full energy contains both ordered orientations of every cross pair,
so \(\cG_\gamma(\cC)\ge2S\), proving
\eqref{eq:orthoplex-left-endpoint-finite-energy-lower-bound}.

For the asymptotic consequence, let \(\cC_\gamma^{\mathsf G}\) minimize
\(\cG_\gamma\) and write
\(\varrho_\gamma=\rho_{\max}(\cC_\gamma^{\mathsf G})\).  Evaluation at
any packing-optimal codebook gives
\(V_\gamma(n+2,n)\le(n+2)(n+1)\).  Hence
\eqref{eq:bounded-energy-packing-consequence}, with
\(\rho^{\ast}_{n+2,n}=0\), gives
\(\varrho_\gamma=O((n\gamma)^{-1})\).  Applying the finite bound to
\(\cC_\gamma^{\mathsf G}\) proves
\eqref{eq:orthoplex-left-endpoint-liminf}.
\end{IEEEproof}

\bibliographystyle{IEEEtran}
\bibliography{bib_spherical_codes_leading_coeff_final}

\end{document}